\documentclass{article}

\usepackage{authblk}
\providecommand{\keywords}[1]{\textbf{\textit{Keywords---}} #1}

\usepackage[usenames,dvipsnames,svgnames,table]{xcolor}
\usepackage{amsmath}
\newtheorem{prop}{Proposition}
\newtheorem{theo}{Theorem}
\newtheorem{definition}{Definition}
\newtheorem{lemma}{Lemma}
\newtheorem{proof}{Proof}
\usepackage{amssymb} 
\usepackage{mathrsfs}
\usepackage{mathtools}

\usepackage[a4paper,left=3cm,right=3cm,top=3cm,bottom=2.5cm]{geometry}

\usepackage{amsmath}
\usepackage{amssymb}
\usepackage{graphicx}
\usepackage{caption}
\usepackage{subcaption}
\usepackage{xcolor} 

\newcommand{\vv}{\mathbf{v}}

\newcommand{\ii}{\mathrm{i}}
\newcommand{\bk}{\mathbf{k}}
\newcommand{\bK}{\mathbf{K}}
\newcommand{\bx}{\mathbf{x}}
\newcommand{\by}{\mathbf{y}}
\newcommand{\bX}{\mathbf{X}}
\newcommand{\EE}{\mathbf{E}}
\newcommand{\bE}{\mathbf{E}}

\newcommand{\bF}{\mathbf{F}}
\newcommand{\DD}{\mathbf{D}}
\newcommand{\HH}{\mathbf{H}}
\newcommand{\bH}{\mathbf{H}}

\newcommand{\BB}{\mathbf{B}}

\newcommand{\LL}{\mathcal{L}}
\renewcommand{\L}{\mathbf{L}}

\newcommand{\RRT}{\mathfrak{R}}
\newcommand{\PPT}{\mathfrak{P}}
\newcommand{\TTT}{\mathfrak{T}}

\newcommand{\MM}{\mathcal{M}}

\newcommand{\vf}{\mathbf{\Psi}}
\newcommand{\bPsi}{\mathbf{\Psi}}
\newcommand{\bPhi}{\mathbf{\Phi}}

\newcommand{\omegaz}{\omega^{0}}

\newcommand{\bu}{\mathbf{u}}

\newcommand{\bPsiso}{\bPsi_{\sigma}^{1}}
\newcommand{\bPsisz}{\bPsi_{\sigma}^{0}}
\newcommand{\pt}{\PPT\TTT}

\makeatletter
\renewcommand*\env@matrix[1][c]{\hskip -\arraycolsep
  \let\@ifnextchar\new@ifnextchar
  \array{*\c@MaxMatrixCols #1}}
\makeatother

\newcommand{\MMo}{\mathcal{M}_{W^1}}

\newcommand{\curl}{\text{curl }}
\newcommand{\la}{\langle}
\newcommand{\ra}{\rangle}
\newcommand{\rai}{\rangle_I}
\renewcommand{\div}{\text{div }}

\DeclareMathOperator{\sech}{sech}
\DeclareMathOperator{\dd}{d}

\graphicspath{{./}{./fig/}}

\begin{document}
\title{Linear and nonlinear electromagnetic waves in modulated honeycomb media}
\author[1]{Pipi Hu}
\author[1]{Liu Hong}
\author[1]{Yi Zhu\thanks{yizhu@tsinghua.edu.cn}}

\affil[1]{Zhou Pei-Yuan Center for Applied Mathematics, Tsinghua, Beijing, 100084, China.}

\maketitle
\begin{abstract}
 Wave dynamics in topological materials has been widely studied recently. A striking feature is the existence of robust and chiral wave propagations that have potential applications in many fields. A common way to realize such wave patterns is to utilize Dirac points which carry topological indices and is supported by the symmetries of the media.  In this work, we investigate these phenomena in photonic media. Starting with Maxwell's equations with a honeycomb material weight as well as the nonlinear Kerr effect, we first prove the existence of Dirac points in the dispersion surfaces of transverse electric and magnetic Maxwell operators under very general assumptions of the material weight. Our assumptions on the material weight are almost the minimal requirements to ensure the existence of Dirac points in a general hexagonal photonic crystal. We then derive the associated wave packet dynamics in the scenario where the honeycomb structure is weakly modulated. It turns out the reduced envelope equation is generally a two-dimensional nonlinear Dirac equation with a spatially varying mass. By studying the reduced envelope equation with a domain-wall-like mass term, we realize the subtle wave motions which are chiral and immune to local defects.  The underlying mechanism is the existence of topologically protected linear line modes, also referred to as edge states. However, we show that these robust linear modes do not survive with nonlinearity. We demonstrate the existence of nonlinear line modes, which can propagate in the nonlinear media based on high-accuracy numerical computations. Moreover, we also report a new type of nonlinear modes which are localized in both directions.
\end{abstract}


\section{Introduction}

The past few years have witnessed an explosion of researches on topological materials in different fields. Many novel and subtle wave dynamics are investigated in these materials. One hallmark is the topological propagation of wave modes which are immune to defects and disorders \cite{cheng2016robust, HR:07,plotnik2013observation,poo2011experimental,wang2009observation}. Among those novelly designed materials, one focus is the honeycomb-based materials, in which the underlying symmetries play essential roles in topological phenomena \cite{ablowitz2010evolution, Fefferman2012Honeycomb, RMP-Graphene:09}. This work is concerned with topological wave dynamics in the nonlinear photonic media.

Maxwell's equations for the electromagnetic fields in nonlinear Kerr media read
\begin{align}\label{eq:maxwellmain}
&\frac{\partial \DD}{\partial t}=\curl \HH, \quad \frac{\partial \BB}{\partial t}=-\curl \mathbf{E},\\
&\div\DD = 0,\quad \div\BB = 0,
\end{align}
for electric field $\EE$ and displacement $\DD$, magnetic field $\HH$ and induction field $\BB$ with the following constitutive relation
\begin{equation}
\DD =\hat{\epsilon} \EE + \sigma|\EE|^2\EE, \quad \BB = \hat{\mu}\HH,
\label{eq.constitutive}
\end{equation}
where all fields $\mathbf{F}=(\bF^{(1)}, \bF^{(2)},\bF^{(3)})^T$ are complex-valued functions, the linear permittivity $\hat{\epsilon}$ and permeability $\hat{\mu}$ are $3\times 3$  positive-definite Hermitian matrices and $\sigma\in \mathbb{R}$ represents the nonlinear Kerr coefficient.

This work focuses on two-dimensional photonic materials with the following specific material weights
\begin{equation}\label{material_weight}
\hat{\epsilon}=
\begin{pmatrix}
    \epsilon_{11} & \epsilon_{12} & 0 \\
    \epsilon_{21} & \epsilon_{22} & 0 \\
    0 & 0 & \epsilon_3 \\
\end{pmatrix},\quad \hat{\mu}= \begin{pmatrix}
    \mu_{11} & \mu_{12} & 0 \\
    \mu_{21} & \mu_{22} & 0 \\
    0 & 0 & \mu_3 \\
\end{pmatrix},
\end{equation}
where $\hat{\epsilon}=\hat{\epsilon}(\bx),~\hat{\mu}=\hat{\mu}(\bx),  ~\bx=(x_1, ~x_2)^T\in \mathbb{R}^2$ vary in the transverse plane and are invariant along the longitudinal direction.  In this setup, the general nonlinear Maxwell's equations are still very complicated. However, if the nonlinear effect is negligible, i.e., $\sigma=0$, Maxwell's equations with linear material weight (\ref{material_weight}) can be simplified greatly. Indeed, the electromagnetic fields can be divided into two decoupled components, the so-called transverse electric (TE) field $\bPsi_e = (\bE^{(2)}, -\bE^{(1)}, \bH^{(3)})^T$ and transverse magnetic (TM) field $\bPsi_m = (\bH^{(2)}, -\bH^{(1)}, \bE^{(3)})^T$, (We switch the first two entries of the TE/TM modes for the purpose of notational simplification. These forms are equivalent to the standard decoupling \cite{joannopoulos2008molding}.) which satisfy
\begin{equation*}
\ii\partial_t
\begin{pmatrix}
\bPsi_e\\
\bPsi_m\\
\end{pmatrix} + \begin{pmatrix}
    \MM_{W_e} & 0\\
    0 & -\MM_{W_m} \\
\end{pmatrix}\begin{pmatrix}
\bPsi_e\\
\bPsi_m\\
\end{pmatrix}
=0.
\end{equation*}
Hereafter, we use the notation $\MM_W:=W(\bx)\LL$ for a given $3\times3$ matrix function $W(\bx)$ and
\begin{equation}\label{eq.LL}
\LL=\begin{pmatrix}
    0 & 0 &\ii\partial_{x_1}\\
    0 & 0 &\ii\partial_{x_2}\\
    \ii\partial_{x_1} & \ii\partial_{x_2} & 0 \\
\end{pmatrix}.
\end{equation}

From (\ref{material_weight}), we focus on material weight matrices $W(\bx)$ of the form

\begin{equation}\label{WeWm}
W_e(\bx)=\begin{pmatrix}
    d_\epsilon^{-1}\epsilon_{11}  & d_\epsilon^{-1}\epsilon_{21} & 0 \\
    d_\epsilon^{-1}\epsilon_{12}  & d_\epsilon^{-1}\epsilon_{22} & 0 \\
    0 & 0 & \mu_3^{-1}\\
\end{pmatrix},
 W_m(\bx)= \begin{pmatrix}
    d_\mu^{-1}\mu_{11} & d_\mu^{-1}\mu_{21}& 0 \\
    d_\mu^{-1}\mu_{12} & d_\mu^{-1}\mu_{22}& 0 \\
    0 & 0&\epsilon_3^{-1}\\
\end{pmatrix},
\end{equation}

where $d_\epsilon=\epsilon_{11}\epsilon_{22}-\epsilon_{12}\epsilon_{21}>0$, and $d_\mu=\mu_{11}\mu_{22}-\mu_{12}\mu_{21}>0$.

Thanks to the unified form, we can study the TE/TM components similarly by solving the following eigenvalue problem
\begin{equation}
    \MM_{W}\bPsi = \omega\bPsi.
\label{eq.eigenvalueo}
\end{equation}
In accordance with (\ref{WeWm}),  the material weight under study is of the specific form
\begin{equation*}
W(\bx)=\begin{pmatrix}
    A(\bx)& \mathbf{0}_{{2\times1}}\\
    \mathbf{0}_{{1\times2}} & a(\bx)\\
\end{pmatrix}.
\end{equation*}

Throughout this paper, we assume that $W=W(\bx)$ is an admissible material weight in the following sense.

\begin{definition}\label{def.adm}
  A $d\times d$ $(d\in\mathbb{Z}_+)$ matrix function $W(\bx)$ is called admissible if it is (1) Hermitian; and (2) elliptic, i.e., there exist $c, C>0$ such that
   for any $\xi\in \mathbb{C}^d$ and $\bx \in \mathbb{R}^2$,
  \begin{equation*}
  c|\xi|^2\le \xi^* W(\bx) \xi\le C|\xi|^2.
  \end{equation*}
\end{definition}

\subsection{Physical motivations and main results}
By designing different architectures, i.e., manipulating the material weight $W(\bx)$, researchers can produce many novel wave propagation patterns \cite{joannopoulos2008molding}. A recent focus is to realize the so-called topological wave propagation in photonic materials \cite{cheng2016robust,HR:07,plotnik2013observation,poo2011experimental}. The existence of topologically protected edge states is the hallmark. Their immunity to defects ensures robust and nearly lossless energy/signal transfers, which have important applications in many different fields. The fast-developing and huge experimental realizations require rigorous analysis from the theoretical aspect.

To generate robust wave modes, one needs to design the materials with certain symmetries. A typical way is to utilize hexagonal lattices. Together with other symmetries, the Dirac points, conically degenerate points in the dispersion surface, can exist. These materials are often referred to as ``honeycomb material". By modulating the honeycomb structure in certain manners, linear topologically protected edge states can be realized. One goal of this work is to rigorously demonstrate the existence of Dirac points under some very general assumptions on the material weight $W(\bx)$. Specifically, we make the following achievements.
\begin{enumerate}
\item we characterize the \emph{honeycomb material weight} associated with Maxwell operator $\MM_{W}$ in Section \ref{sec.fb}. To this end, we introduce proper function spaces $\L_\bk^2(\Lambda)$, time-reversal $\TTT$, parity-inversion $\PPT$ and $\frac{2\pi}{3}$-rotation $\RRT$ symmetries. Definition \ref{honeycomb_material} gives the essential ingredients of the honeycomb material weight which guarantees the existence of Dirac points for Maxwell operator $\MM_{W}$. Proposition \ref{prop.honeychara} characterizes the honeycomb material weight. To the best of our knowledge, this has not been well described before this work.
\item We demonstrate the existence of Dirac points. After discussing the structure of eigenspaces of $\MM_{W}$ in $\L_\bK^2$, we rigorously prove there exists a conically degenerate point, a.k.a., Dirac point, in the dispersion surfaces of $\MM_{W}$ at the high symmetry point $\bK$ with very general assumptions. Many novel topological phenomena are associated with Dirac points. Our rigorous analysis clarifies what the minimal requirements are needed to obtain Dirac points in a very general photonic setup.  Moreover, we also show that a $\PPT\TTT$-symmetry breaking perturbation to the honeycomb material weight leads to the disappearance of Dirac points and local spectral gap opening. A simple example of the honeycomb material weight and the associated dispersion surfaces are illustrated. The numerical simulations agree very well with our analysis.
\end{enumerate}

After establishing the local analytical structure of the spectrum of $\MM_{W}$ in the vicinity of Dirac points, we can achieve the other goal of this work, i.e., the nonlinear dynamics of the wave packet associated with Dirac points with a slowly modulated honeycomb media. Our results are summarized as follows.
\begin{enumerate}
\setcounter{enumi}{2}
\item We derive the nonlinear envelope equation in a most interesting parameter regime, where the envelope scale, nonlinearity, and material weight modulation are maximally balanced. By implementing a multi-scale analysis, we obtain the envelope equation, which is a nonlinear Dirac equation with a spatially varying mass. The reduction is directly applied to Maxwell's equations (\ref{eq:maxwellmain})-(\ref{eq.constitutive}), and it includes many nontrivial computations of the solvability conditions using symmetry arguments. To the best of our knowledge, this reduction from the nonlinear Maxwell's equation has not been implemented in the literature.

\item By analyzing the reduced envelope equation, we explain many interesting wave propagation patterns. Both the analysis and numerical simulations of the original Maxwell's equations are very tough due to multi-scale features of the physical problems, and some underlying mechanisms are buried in the complicated structure. With the simple form of the envelope equation with a domain-wall mass, we first show the existence of linear line modes, a.k.a, topologically protected edge states, and their chiral and robust features. We also show that these striking features break down when nonlinearity is included in the system. Via numerical methods, we obtain the nonlinear line modes as well as fully localized nonlinear lump-like solitons.
\end{enumerate}

\subsection{Connections to previous studies and outline}

In the current study of topological wave propagations, the Dirac point, which carries topological characterizations \cite{RevModPhys.82.3045,qi2011topological}, is frequently used to realized subtle phenomena. Hexagonal periodicity is the best candidate due to their symmetries. However, more conditions are needed to guarantee the existence of Dirac points. For a physical system, it is a key problem that what minimal conditions are needed to ensure the existence of Dirac points. Despite a large amount of numerical and/or asymptotic studies on the existence of Dirac points, see for example \cite{ablowitz2012nonlinear,ablowitz2013nonlinear,guo2019bloch,wallace1947band,ablowitz2011nonlinear},  Fefferman and Weinstein first gave the general condition for the honeycomb potential under which they rigorously proved the existence of Dirac points for the Schr\"odinger operator \cite{fefferman2018honeycomb}. This work solved a long-standing open problem in quantum mechanics.  Later, one author of this work Zhu together with Lee-Thorp and Weinstein extended the conditions of the honeycomb media, and the rigorous demonstration of Dirac points to a 2nd order elliptic operator with a divergence form \cite{lee2019elliptic}. This elliptic operator can be used to describe some special TE mode of electromagnetic waves in a photonic crystal. Our current work aims to extend the analysis to a general photonic material. To do so, we clarify the symmetries, function spaces, general requirements to prove the existence of Dirac points. Our extension provides the most general theory for a honeycomb photonic material. In addition, compared to previous studies in which the operators are all scaler operators, we deal with Maxwell operator directly, which is a vector operator. The extension itself involves nontrivial arguments and calculations.

The other part of this work is devoted to the nonlinear dynamics of wave packets associated with Dirac points in slowly modulated honeycomb media. We show that the reduced equation is the two-dimensional nonlinear Dirac equation with a spatially varying mass. We want to point out that the reduced envelope equation is very similar to the nonlinear coupled-mode equation in which the mass is a constant  \cite{dohnal2013coupled}, see also \cite{curtis2015dynamics}. It is the spatially varying mass term that makes the essential difference. Indeed, with a domain-wall-like mass, our envelope equation can capture the topologically protected wave propagation. Our analysis and numerical results in  \ref{sec.dynamics} show such powerful capabilities. In the literature, this importance of the linear Dirac equation with a varying mass has been noticed \cite{bal2017topological,bal2018continuous,xie2019wave}. In the past few years, people began to use similar nonlinear models to describe nonlinear topological modes \cite{smirnova2019topological}. The direct reduction of the envelope equation from Maxwell's equations with a modulated honeycomb material weight has not been done before our work. For example, the nonlinear terms in \cite{smirnova2019topological} is added artificially without any reasonable explanations. Our current work provides a complete and consistent theory, including the conditions for the existence of Dirac points and the associated linear spectrum of Maxwell's operator, the reduction of the nonlinear envelope equation. All terms and coefficients have very clear physical interpretations. On the other hand, our derivation has a lot of nontrivial symmetry arguments in the detailed calculations when dealing with the modulation and nonlinearity. In addition, the linear Dirac equation with a sign-changed mass has been used to describe topological quantum mechanics. In the literature, most of the works focus on the topologically protected edge states and topological invariants \cite{bal2017topological,bal2018continuous}. With our reduction, we can study how the nonlinearity affects the topological wave propagation. Specially, we generate the nonlinear line modes and fully localized solitary waves corresponding to the reduced envelope equation.

The rest of the paper is organized as follows. In Section \ref{sec.fb}, Floquet-Bloch theory for the Maxwell operator $\MM_W$ is reviewed. We then characterize the required symmetries and define the honeycomb material weight. The high-symmetry points and properties of associated function spaces are discussed. In Section \ref{sec.dirac}, after investigating the eigenspaces of $\MM_W$ at high-symmetry points, we prove the existence of Dirac points in the dispersion surfaces of $\MM_W$ with a honeycomb material weight that we define in the last section. We then study the persistence of the Dirac points under perturbations. Numerical simulations on the spectrum of $\MM_W$ with a typical honeycomb material weight are presented and show the agreements with our analysis. The derivation of the nonlinear envelope equation is implemented in Section \ref{sec.envelope}.
In Section \ref{sec.dynamics}, we analyze and numerically simulate the envelope equations in different scenarios. Linear and nonlinear line modes are obtained, and their typical dynamics are studied. We also show a new-type of lump-like solitary wave solutions to the reduced nonlinear envelope equation. We conclude our results in Section \ref{sec.conclusion}.

\section{Honeycomb media and Floquet-Bloch theory}\label{sec.fb}
In this section, we list the function spaces, honeycomb structures, symmetries, Floquet-Bloch theory, and some preliminary results which are required for our analysis in the next sections.
\subsection{Hexagonal lattice}

A hexagonal lattice $\Lambda = \mathbb{Z}\vv_1\oplus\mathbb{Z}\vv_2$ is generated by the lattice vectors
$    \vv_1 = (
        \frac{\sqrt{3}}{2},
        \frac{1}{2}
    )^T$,  and $\vv_2 = (
        \frac{\sqrt{3}}{2},
        -\frac{1}{2}
    )^T.
$
Here we have normalized the lattice vectors for simplicity.

\begin{figure}[htbp]
     \centering
     \begin{subfigure}[b]{0.4\textwidth}
         \centering
         \label{figh:a}\includegraphics[width=\textwidth]{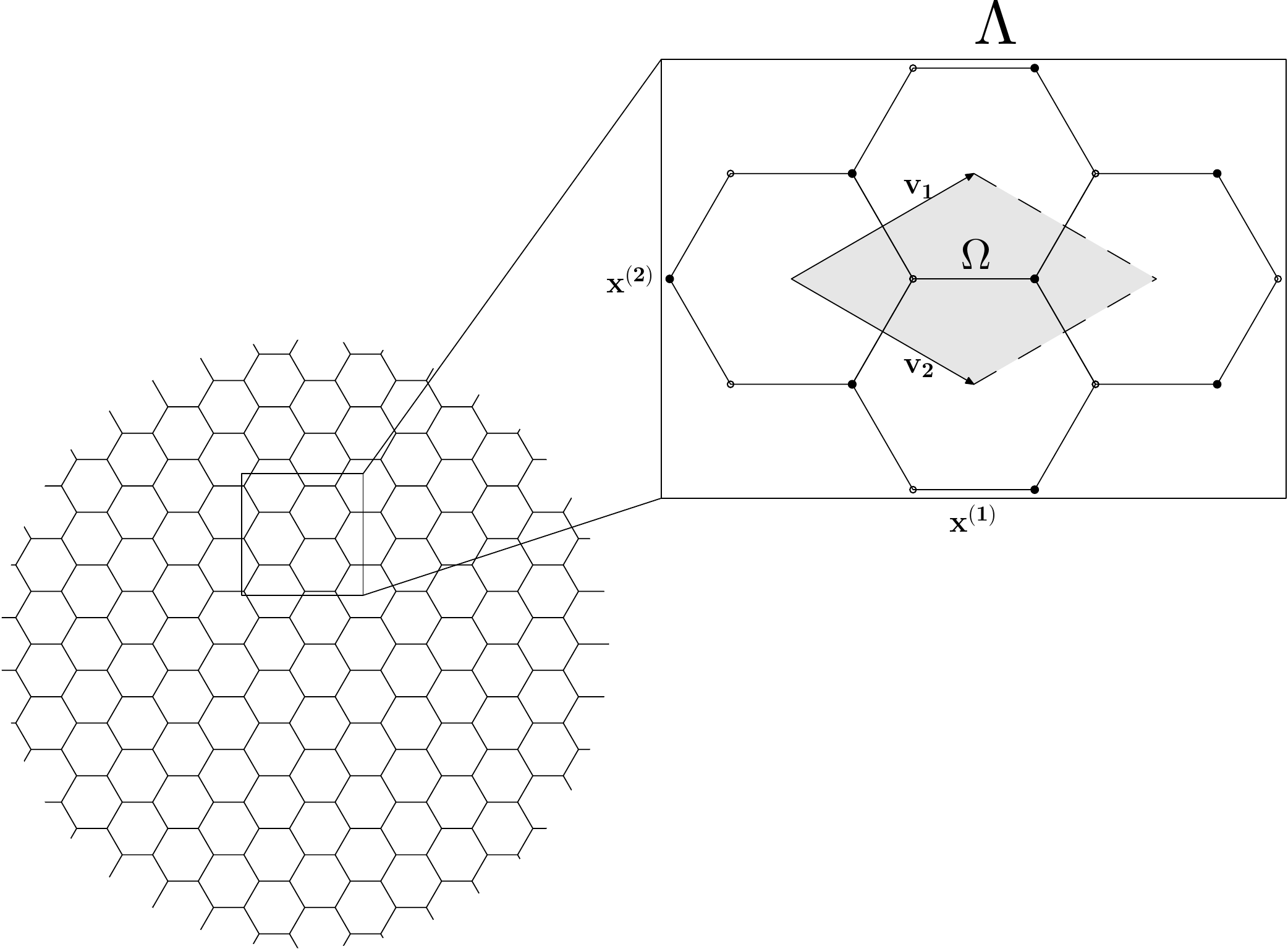}
         \end{subfigure}
     \hspace{1em}
          \begin{subfigure}[b]{0.18\textwidth}
         \centering
         \label{figh:b}\includegraphics[width=\textwidth]{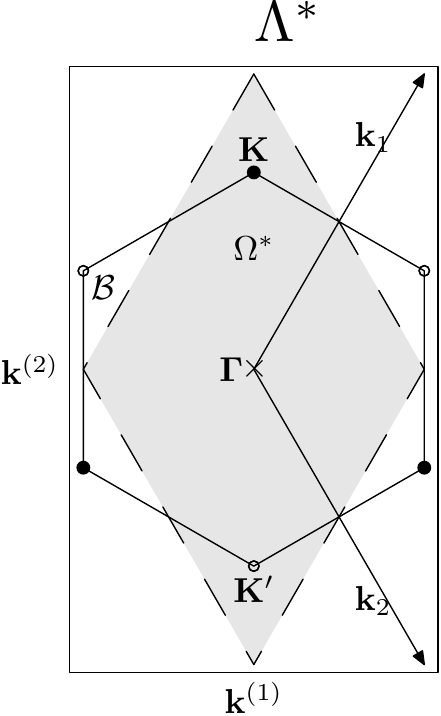}
     \end{subfigure}
       \caption{Left Panel: hexagonal lattice, lattice vectors $\vv_j, ~j=1,2$ and the unit cell $\Omega$ (the shadow region). Right Panel: the dual fundamental cell $\Omega^*$ (the shadow region) and the Brillouin zone $\mathcal{B}$ (the region surrounded by the hexagon). $\bK$ and $\bK'$ are labeled. }
  \label{fig.hexagon}
\end{figure}

The fundamental cell is chosen to be
\begin{equation*}
\Omega=\{\theta_1\vv_1+\theta_2\vv_2\,|\, 0\leq \theta_j<1,\, j=1,2 \}.
\end{equation*}

The dual lattice is $
    \Lambda^* = \mathbb{Z}\bk_1\oplus\mathbb{Z}\bk_2$
with the dual lattice vectors
$    \bk_1 =\frac{4\pi}{\sqrt{3}} (
        \frac{1}{2},
        \frac{\sqrt{3}}{2}
    )^T$,  and
        $\bk_2 = \frac{4\pi}{\sqrt{3}}(
        \frac{1}{2}
        -\frac{\sqrt{3}}{2}
        )^T,
$
satisfying the reciprocal relations $\bk_i\cdot \vv_j=2\pi\delta_{ij},~ i,j=1,2$. Throughout this paper, we choose the fundamental dual cell to be
$\Omega^*=\{\theta_1\bk_1+\theta_2\bk_2\,|\, -\frac12\leq \theta_j<\frac12,\, j=1,2 \}$.
It is remarked that this choice of the fundamental cell is equivalent to the Brillouin zone $\mathcal{B}$ that is frequently used in physical literature. The physical and dual lattices, as well as their fundamental cells are shown in Figure \ref{fig.hexagon}.

For the purpose of this work, we introduce the following spaces for 3-tuple vector functions defined in $\mathbb{R}^2$

\begin{align*}
\L_\text{per}^2(\Lambda)=\bigl\{\bu(\bx)=\bigl(\bu^{(1)}(\bx),\bu^{(2)}(\bx),\bu^{(3)}(\bx)\bigr)^T\,|\, \bu^{(j)}(\bx)\in L_\text{per}^2(\Lambda),\, j=1,2,3\bigr\},
\end{align*}

and

\begin{align*}
\L_\bk^2(\Lambda)=\bigl\{\bPsi(\bx)=\bigl(\bPsi^{(1)}(\bx),\bPsi^{(2)}(\bx),\bPsi^{(3)}(\bx)\bigr)^T\,|\, \bPsi^{(j)}(\bx)\in L_\bk^2(\Lambda),\, j=1,2,3\bigr\},
\end{align*}

where $L_\text{per}^2(\Lambda)$ and $L_\bk^2(\Lambda)$ are standard square integrable function spaces for periodic and $\bk$-quasiperiodic functions, i.e.,
$L^2_\text{per}(\Lambda) = \{f(\bx)\in L^2_{\text{loc}}(\mathbb{R}^2,\mathbb{C})| f(\bx+\vv)=f(\bx), \vv\in\Lambda\}$, and
$L^2_\bk(\Lambda) = \{f(\bx)|e^{-\mathfrak{i}\bk\cdot\bx}f(\bx)\in L^2_\text{per}(\Lambda)\}.$

For a material weight $W(\bx)$, we define the following weighted inner product in $\L^2_{\text{per}}(\Lambda)$
\begin{equation}\label{eq.Wweight}
    \langle\bPsi_1,\mathbf{\bPsi_2}\rangle_{_W} = \int_{\Omega}\bPsi_1^*(\bx) W^{-1}(\bx)\bPsi_2(\bx)\dd\bx,
\end{equation}
where the superscript asterisk ``*"  represents the conjugate transpose. Since $W(\bx)$ is elliptic, the weighted norm $\| \cdot \|_W$ induced by this inner product is equivalent to the standard norm $\|\cdot \|_{\L^2_{\text{per}}(\Lambda)}$. Note that $\bPhi_1W^{-1}(\bx)\bPhi_2$ is in $L^2_\text{per}(\Lambda)$ for any $\bPhi_j(\bx)\in\L^2_\bk(\Lambda)$, $j=1,2$. So the inner product can be extended to $\L^2_\bk(\Lambda)$ for any $\bk\in\mathbb{R}^2$.

\subsection{Floquet-Bloch Theory}\label{sec.bloch}
In this subsection, we briefly review the Floquet-Bloch theory for the operator $\MM_{W}$ when $W(\bx)$ is $\Lambda$-periodic, see for example \cite{joannopoulos2008molding,lee2019elliptic,De_Nittis-Lein:14}. The spectrum of $\MM_{W}$ can be obtained by solving the following $\L^2_\bk(\Lambda)$-eigenvalue problem
\begin{equation}\label{eq.bloch}
    \mathcal{M}{_W} \bPsi(\bx) = \omega\bPsi(\bx),\quad \bPsi(\bx)\in \L^2_\bk(\Lambda).
\end{equation}

Alternatively, by setting
$    \bPsi(\bx) = e^{\mathfrak{i}\bk\cdot\bx}\mathbf{u}(\bx)$,
we transform the eigenvalue problem (\ref{eq.bloch}) to
\begin{equation}
    \MM_{W}(\bk)\bu(\bx) = \omega(\bk)\bu(\bx),\quad \bu(\bx)\in\L^2_\text{per}(\Lambda),
\label{eq:eigenk}
\end{equation}
where

\begin{equation*}
    \MM_{W}(\bk) = e^{-\mathfrak{i}\bk\cdot\bx}\MM_{W} e^{\mathfrak{i}\bk\cdot\bx}
    =W(\bx)\begin{pmatrix}
    0 & 0 &\ii\partial_{x_1}-\bk^{(1)}\\
    0 & 0 &\ii\partial_{x_2}-\bk^{(2)}\\
    \ii\partial_{x_1}-\bk^{(1)} & \ii\partial_{x_2}-\bk^{(2)} & 0 \\
\end{pmatrix},
\end{equation*}

 and $\bk=(\bk^{(1)},\bk^{(2)})^T$.

For fixed $\bk\in\Omega^*$, the eigenvalue problem (\ref{eq:eigenk}) has a series of discrete eigenvalues
\begin{equation*}
    \cdots \leq \omega_{-b}(\bk) \leq \cdots \leq \omega_{-1}(\bk) \leq \omega_{0}(\bk)\leq \omega_1(\bk) \leq \cdots \leq \omega_b(\bk) \leq \cdots .
\end{equation*}
The mappings $\bk  \rightarrow \omega(\bk)$ are called band dispersion functions which are Lipschitz continuous.

\subsection{$\RRT$,\,$\PPT$ and $\TTT$ symmetries}

Symmetry plays a very important role in the understanding of physical phenomena.
We introduce the following symmetry operators.
For any $\bPsi(\bx)=\bigl(\bPsi^{(1)}(\bx),\bPsi^{(2)}(\bx),\bPsi^{(3)}(\bx)\bigr)^T$ defined in $\mathbb{R}^2$, we define  the parity inversion operator $\PPT$ as
\begin{equation}\label{def.P}
    [\PPT\vf](\bx) =
    \begin{pmatrix}
        -I_{2\times 2}&\mathbf{0}\\
        \mathbf{0}&1\\
    \end{pmatrix}
    \vf(-\bx),
\end{equation}
the time reversal operator $\TTT$ as
\begin{equation}\label{def.T}
    [\TTT\vf](\bx) =
    \begin{pmatrix}
        -I_{2\times 2}&\mathbf{0}\\
        \mathbf{0}&1\\
    \end{pmatrix}
    \overline{\vf(\bx)},
\end{equation}
and the $\frac{2\pi}{3}$-rotation operator $\RRT$ as
\begin{equation}\label{def.R}
    [\RRT\vf](\bx) =
     \begin{pmatrix}
        R&\mathbf{0}_{2\times1}\\
        \mathbf{0}_{1\times2}&1\\
    \end{pmatrix}
    \vf(R^*\bx),
\end{equation}
where the $2\times 2$ matrix $R$ is a clockwise $\frac{2\pi}{3}$-rotation matrix $R = \bigl(\begin{smallmatrix}
        -\frac{1}{2} &\frac{\sqrt{3}}{2}\\
        -\frac{\sqrt{3}}{2} &-\frac{1}{2}\\
    \end{smallmatrix}\bigr)$.

Remark that the eigenvalues of $R$ are $\tau=e^{\ii\frac{2\pi}{3}}$ and $\overline{\tau}=e^{-\ii\frac{2\pi}{3}}$, with the corresponding eigenvectors  $\frac{1}{\sqrt{2}}\begin{pmatrix}1\\\ii\\ \end{pmatrix}$ and $\frac{1}{\sqrt{2}}\begin{pmatrix}1\\-\ii\\ \end{pmatrix}$.

We use the terminology that $\vf(\bx)$ is $\RRT$ invariant if $[\RRT \vf](\bx) = \vf(\bx)$. Similar terminologies are used for $\PPT-$invariance and $\TTT-$invariance. We also use  $\PPT\TTT$ as the compound of operators $\PPT$ and $\TTT$.

\subsection{Honeycomb material weight}
This work focuses on the photonic material with honeycomb structures defined as follows.
\begin{definition}\label{honeycomb_material} A $3\times3$ matrix function $ W(\bx)=\bigl(\begin{smallmatrix}A(\bx)&\mathbf{0}\\\mathbf{0}&a(\bx)\end{smallmatrix}\bigr)$
is a honeycomb material weight if it is admissible in the sense of Definition \ref{def.adm} and further satisfies
\begin{enumerate}
\item $A(R^*\bx)=R^*A(\bx)R$, and $a(R^*\bx)=a(\bx)$;
\item $\overline{A(-\bx)}=A(\bx)$, and $a(\bx)$ is real and even.
\end{enumerate}
\end{definition}

Evidently, the identity matrix $I_{{3\times3}}$ is a trivial example of the honeycomb material weight. Moreover, $W^{-1}(\bx)$ is a honeycomb material weight if $W(\bx)$ is.

An obvious advantage of the honeycomb material weight is given in the following proposition.

\begin{prop}\label{prop.honeychara}
    The following commutators vanish if $W(\bx)$ is a honeycomb material weight:
\begin{equation}
[\MM_{W}, \RRT]=0, \quad \text{and}\quad [\MM_{W}, \PPT\TTT]=0.
\end{equation}
\end{prop}

\subsection{Properties of high symmetry points}
It will be seen that symmetries play a very important role on the existence of degenerate Dirac points since the honeycomb material weight is described through the symmetries, see Definition \ref{honeycomb_material} and Proposition \ref{prop.honeychara}.  Apparently, for any $\bk\in\Omega^*$, $\bPsi \in \L^2_\bk(\Lambda)$ if and only if  $\PPT\TTT \bPsi\in \L^2_\bk(\Lambda)$. Further $\bPsi \in \L^2_\bk(\Lambda)$ if and only if  $\RRT \bPsi\in \L^2_{<R\bk>}(\Lambda)$ where $<R\bk>$ is the representative of $R\bk$ in $\Omega^*$, i.e., there exists $m_1,m_2\in \mathbb{Z}$ such that $<R\bk>=R\bk+m_1\bk_1+m_2\bk_2 \in \Omega^*$. Therefore
$$\PPT\TTT\L^2_\bk(\Lambda)=\L^2_\bk(\Lambda),$$
but it does not hold for the rotational operator $\RRT$ if $<R\bk>\neq\bk$.

For a given $\bk\in \Omega^*$, there uniquely exist $p_j\in [-\frac{1}{2},\frac{1}{2}),~ j=1,2$ such that $\bk = p_1\bk_1+p_2\bk_2$.
Recalling $R\bk_1=\bk_2,~ R\bk_2=-\bk_1-\bk_2$, we have
\[
R\bk -\bk =-(p_1+p_2)\bk_1 +(p_1-2p_2)\bk_2.
\]

In order that $<R\bk>=\bk$, there must exist $m_j\in \mathbb{Z},~j=1,2$ such that
$$
\bk=R\bk+m_1\bk_1 +m_2\bk_2.
$$

Thus, both $p_1+p_2$ and $p_1-2p_2$ are integers. Apparently, there are three solutions $(p_1,p_2)=(0,0),(\frac{1}{3},-\frac{1}{3}),(-\frac{1}{3},\frac{1}{3})$.
Denote
$\Gamma=\mathbf{0}$, $\bK = \frac{1}{3}(\bk_1-\bk_2)$, and $\bK'=-\bK.$
The three points are referred as the high symmetry points with respect to operator $\RRT$.

In this work, we focus on the $\bK$ point, while the analysis for $\bK'$ is the same. The following properties about the space $\L^2_{\bK}$ are frequently used in our later analysis. Hereafter, we frequently suppress the lattice symbol $\Lambda$ for simplicity, i.e., $\L^2_{\bK}=\L^2_{\bK}(\Lambda)$.
\begin{lemma}    \label{prop.inner}
Let  $W(\bx)$ be a honeycomb material weight. The following identities hold for any $\bPsi_j \in \L^2_{\bK},~ j=1,2$:
\begin{equation*}
        \la\RRT\bPsi_1,\RRT\bPsi_2\ra_{_W}=\la\bPsi_1,\bPsi_2\ra_{_W}\quad
\text{and }
    \la\PPT\TTT\bPsi_1,\PPT\TTT\bPsi_2\ra_{_W}=\overline{\la\bPsi_1,\bPsi_2\ra}_{_W}.
\end{equation*}
\end{lemma}

\begin{proof}
By definition of $\RRT$ in (\ref{def.R}), we directly calculate

\begin{equation*}
    \begin{array}{ll}
    \la\RRT\bPsi_1,\RRT\bPsi_2\ra_{_W} =&  \displaystyle\int_{\Omega}\bPsi_1^*(R^*\bx)\begin{pmatrix}
        R^* & 0\\
        0& 1\\
    \end{pmatrix} W^{-1}(\bx) \begin{pmatrix}
        R & 0\\
        0 & 1\\
    \end{pmatrix}\bPsi_2(R^*\bx)\dd \bx\\[0.5cm]
    \overset{\mathbf{y} = R^*\bx}{=}&
    \displaystyle\int_{R^*\Omega}\bPsi_1^*(\mathbf{y})\begin{pmatrix}
        R^* & 0\\
        0 & 1\\
    \end{pmatrix}\begin{pmatrix}
        A^{-1}(R\mathbf{y}) & 0\\
        0 & a^{-1}(R\mathbf{y})\\
    \end{pmatrix} \begin{pmatrix}
        R & 0\\
        0 & 1\\
    \end{pmatrix}\bPsi_2(\mathbf{y})\dd \mathbf{y}.
    \end{array}
\end{equation*}

$A(R^*\bx)=R^*A(\bx)R$ by Definition \ref{honeycomb_material} implies
\begin{equation*}
R^*A^{-1}(R\mathbf{y})R=(R^*A(R\mathbf{y})R)^{-1}=(A(R^*R\mathbf{y}))^{-1}=A^{-1}(\mathbf{y}).
\end{equation*}
It follows that
\begin{equation*}
    \la\RRT\bPsi_1,\RRT\bPsi_2\ra_{_W} = \int_{\Omega}\bPsi_1^*(\mathbf{y})\begin{pmatrix}
        A^{-1}(\mathbf{y}) & 0\\
        0 & a^{-1}(\mathbf{y})\\
    \end{pmatrix}\bPsi_2(\mathbf{y})\,\dd\mathbf{y} = \la\bPsi_1,\bPsi_2\ra_{_W}.
\end{equation*}
Similarly, we have
$\la\PPT\TTT\bPsi_1,\PPT\TTT\bPsi_2\ra_{_W} = \overline{\la\bPsi_1,\bPsi_2\ra}_{_W}$.
\end{proof}

Noting that $\RRT: \L^2_{\bK}\rightarrow\L^2_{\bK}$ is isometry and $\RRT^3=I_d$, we can divide $\L^2_{\bK}$ into a direct sum of the eigenspaces of $\RRT$. Namely,
\begin{equation}\label{directSum}
    \L^2_{\bK} =  \L^2_{\bK,1}\oplus \L^2_{\bK,\tau}\oplus \L^2_{\bK,\overline{\tau}},
\end{equation}
where
\begin{equation}\label{def.subspace}
\L^2_{\bK,\sigma}=\left\{\vf(\bx)\in\L^2_{\bK}~|~ [\RRT\vf](\bx)=\sigma \vf(\bx)  \right\}, \quad \sigma=1,~\tau,~\overline{\tau}.
\end{equation}
The direct sum (\ref{directSum}) is actually an orthogonal sum by Lemma \ref{prop.inner}.

\begin{lemma}\label{lemma.f}
Define the mapping $\mathcal{F}:\L^2_\bK\times\L^2_\bK\rightarrow \mathbb{C}^2$
    \begin{equation}
        \mathcal{F}(\bPsi,\bPhi):= \int_{\Omega}\bigl[\overline{\bPsi}^{\bot}(\bx)\bPhi^{(3)}(\bx)+\overline{\bPsi}^{(3)}(\bx)\bPhi^{\bot}(\bx)\bigr] \dd\bx, \quad \bPsi,\bPhi\in \L^2_{\bK}.
        \label{eq.deff}
    \end{equation}
Then
$\mathcal{F}(\cdot,\cdot)$ satisfies the following properties:
\begin{enumerate}
    \item $\mathcal{F}(\cdot,\cdot)$ is sesquilinear;
    \item $\mathcal{F}(\bPsi,\bPhi)=\overline{\mathcal{F}(\bPhi,\bPsi)}$;
    \item  $\mathcal{F}(\RRT\bPsi,\RRT\bPhi)=R\mathcal{F}(\bPsi,\bPhi)$.
\end{enumerate}
\end{lemma}

 \begin{proof}
 The first two properties are evident from the definition. We now prove the third one.

 Recall the definition of $\RRT$ in (\ref{def.R}), $[\RRT\bPsi](\bx) = \begin{pmatrix}
         R\bPsi^{\bot}(R^*\bx)\\
         \bPsi^{(3)}(R^*\bx)\\
     \end{pmatrix}$.
A direct calculation shows
     \begin{equation*}
         \begin{split}
             \mathcal{F}(\RRT\bPsi,\RRT\bPhi) &=\int_{\Omega}\bigl[R\overline{\bPsi}^{\bot}(R^*\bx)\bPhi^{(3)}(R^*\bx)+\overline{\bPsi}^{(3)}(R^*\bx)R\bPhi^{\bot}(R^*\bx)\bigr]\dd\bx\\[0.3cm]
             &\overset{\mathbf{y} = R^*\bx}{=} \int_{R^*\Omega}R\bigl[\overline{\bPsi}^{\bot}(\by )\bPhi^{(3)}(\by)+\overline{\bPsi}^{(3)}(\by)\bPhi^{\bot}(\by)\bigr]\dd\by\\[0.3cm]
             &=R\mathcal{F}(\bPsi,\bPhi).
         \end{split}
     \end{equation*}
\end{proof}

Interestingly, $\mathcal{F}(\bPsi,\bPhi)$ has some specific directions in $\mathbb{C}^2$ if $\bPsi$ and $\bPhi$ are choosen in the above subspaces of $\L^2_\bK$. Namely, we shall prove the following proposition.
\begin{prop}
    Let $\mathcal{F}(\cdot,\cdot)$ be defined in Lemma \ref{lemma.f} and $\bPsi_j\in\L_{\bK,j}$, $j\in \{1,\tau,\overline{\tau}\}$. There exist complex constants $C_{j,k}$ such that
\begin{equation*}
        \mathcal{F}(\bPsi_j,\bPsi_k)=
        \begin{cases}
            C_{j,k}\begin{pmatrix}
            1\\
            \ii\\
        \end{pmatrix} & \text{if}\,\, (j,k)=(1,\tau)\, , (\tau,\overline{\tau})\, , \text{or}\,
     (\overline{\tau},1)\\[0.5cm]
        C_{j,k}\begin{pmatrix}
            1\\
            -\ii\\
        \end{pmatrix} & \text{if}\,\, (j,k)=(\tau,1)\, , (\overline{\tau},\tau)\, , \text{or}\,
         (1,\overline{\tau})\\[0.3cm]
        0& \text{if}\,\, j=k\\
        \end{cases}\quad .
\end{equation*}
\label{prop.f}
\end{prop}

 \begin{proof}
 By (\ref{def.subspace}) and Lemma \ref{lemma.f}, we have
 \begin{equation*}
     R\mathcal{F}(\bPsi_j,\bPsi_k)=\mathcal{F}(\RRT\bPsi_j,\RRT\bPsi_k)= \mathcal{F}(j\bPsi_j,k\bPsi_k)=\overline{j}k\mathcal{F}(\bPsi_j,\bPsi_k).
 \end{equation*}
 As $j,k\in \{1,\tau,\overline{\tau}\}$, $\overline{j}k$ must take values in $\{1,\tau,\overline{\tau}\}$.
   Recall that $R$ only has two eigenvalues $\tau$ and $\overline{\tau}$ with eigenvectors $\frac{1}{\sqrt{2}}\begin{pmatrix}1\\\ii\\ \end{pmatrix}$ and $\frac{1}{\sqrt{2}}\begin{pmatrix}1\\-\ii\\ \end{pmatrix}$. Therefore, $\mathcal{F}(\bPsi_j,\bPsi_k)$ must be $0$ if $\overline{j}k =1$, i.e., $j=k$. In contrast, if $j\neq k$, then $\mathcal{F}(\bPsi_j,\bPsi_k)$ is the eigenvector of $R$ corresponding to the eigenvalue $\overline{j}k$.
A simple enumeration completes the proof.
 \end{proof}

\section{Conically degenerate points}\label{sec.dirac}
With the preparations in the last section, we now turn to the study of the spectrum of the operator $\MM_W$ where $W=W(\bx)$ is a honeycomb material weight. Although the structure of the whole spectrum is impossible to obtain analytically, we can still get the local structure of the dispersion surfaces around high symmetry points such as the $\bK$ point. In this section, we shall show that there exist conically degenerate points in the dispersion surfaces, which are referred to as Dirac points. With such local structure, it is enough for us to study the envelope dynamics associated with the Dirac points.

\subsection{Existence of Dirac points}
The conically degenerate points that we seek are the eigenvalues of the operator $\MM_W$ with multiplicity 2. Before proceeding, we first investigate the properties of the eigenspace of $\MM_W(\bK)$ at the high symmetry point $\bK$ with multiplicity 2. The results are concluded in the following proposition.

\begin{prop}\label{pro.psi}
Let  $W(\bx)$ be a honeycomb material weight in Definition \ref{honeycomb_material}. Assume $\omega_*$ is a two-fold degenerate eigenvalue of $\MM_W$ in $\L^2_\bK$, i.e., the corresponding eigenspance $\mathcal{E}_{\omega_*}$ is two-dimensional. Then either of the following two statements holds
\begin{enumerate}
\item $\mathcal{E}_{\omega_*}\subset \L^2_{\bK,1}$;
\item $\mathcal{E}_{\omega_*}\subset  \L^2_{\bK,\tau}\oplus \L^2_{\bK,\overline{\tau}}$.
\end{enumerate}

Moreover, if the latter case holds, then there exist constant $C_{\omega_*}\geq 0$, $\bPsi_1(\bx)\in  \L^2_{\bK,\tau}$ and $\bPsi_2(\bx)= [\PPT\TTT\bPsi_1](\bx) \in \L^2_{\bK,\overline{\tau}}$ satisfying $\|\bPsi_j(\bx)\|_{W} =1,\, j=1,2$ and $\mathcal{F}(\bPsi_1,\bPsi_2)=C_{\omega_*}\begin{pmatrix}
            1\\
            \ii\\
        \end{pmatrix}$ such that $\mathcal{E}_{\omega_*}=\text{span}\{\bPsi_1,\bPsi_2\}$.

\end{prop}

\begin{proof}
Recall that $\L^2_\bK=\L^2_{\bK,1}\oplus\L^2_{\bK,\tau}\oplus\L^2_{\bK,\overline{\tau}}$ and the direct sum is an orthogonal sum. Assume the 2-dimensional eigenspace $\mathcal{E}_{\omega^*}$ is neither in $\L^2_{\bK,1}$ nor in $\L^2_{\bK,\tau}\oplus\L^2_{\bK,\overline{\tau}}$. Namely, there exists $\bPsi\neq 0\in\mathcal{E}_{\omega^*}$, but $\bPsi\not\in \L^2_{\bK,1}$ and $\bPsi\not\in \L^2_{\bK,\tau}\oplus\L^2_{\bK,\overline{\tau}}$. By $[\MM_{W}, \RRT]=0$, $\RRT\bPsi\in\mathcal{E}_{\omega_*}$ and thus $\RRT\bPsi-\bPsi\in\mathcal{E}_{\omega_*}$.

Let $\bPhi=\RRT\bPsi-\bPsi$. Evidently, $\bPhi\in\L^2_{\bK,\tau}\oplus\L^2_{\bK,\overline{\tau}}$ and $\bPhi\neq 0$ by assumption $\bPsi\not\in\L^2_{\bK,1}$.
Note that $\RRT\bPhi-\tau\bPhi\in\L^2_{\bK,\overline{\tau}}$ and $\RRT\bPhi-\overline{\tau}\bPhi\in\L^2_{\bK,\tau}$. They  belong to $\mathcal{E}_{\omega_*}$ and can not be both zero. Without loss of generality, suppose $\RRT\bPhi-\tau\bPhi\neq0$. $\PPT\TTT(\RRT\bPhi-\tau\bPhi)\in \mathcal{E}_{\omega_*}$ by $[\MM_{W}, \PPT\TTT]=0$. Thus, we have constructed three linearly independent nonzero functions $\bPsi$, $\RRT\bPhi-\tau\bPhi$ and $\RRT\TTT(\RRT\bPhi-\tau\bPhi)$ which violates the condition $\text{dim}{\mathcal{E}_{\omega_*}}=2$. Thus, $\mathcal{E}_{\omega_*}$ is either in $\L^2_{\bK,1}$ or in $\L^2_{\bK,\tau}\oplus\L^2_{\bK,\overline{\tau}}$. We turn to the second part.

Similar to the above argument, there exist $\bPhi_1\neq 0\in\L^2_{\bK,\tau}$, $\bPhi_2=\PPT\TTT\bPhi_1\in\L^2_{\bK,\overline{\tau}}$ and $\|\bPhi_1\|_W=\|\bPhi_2\|_W=1$ such that $\mathcal{E}_{\omega_*}=\text{span}\{\bPhi_1,\bPhi_2\}$. By Propostion \ref{prop.f}, there exists a constant $C_0\in\mathbb{C}$ such that $\mathcal{F}(\bPhi_1,\bPhi_2)=C_0\begin{pmatrix}1\\\ii\\\end{pmatrix}$.
Let $\bPsi_1 = \bPhi_1e^{\ii\arg{C_0}/2}\in\L^2_{\bK,\tau}$, $C_{\omega_*}=|C_0|$ and $\bPsi_2=\PPT\TTT\bPsi_1\in\L^2_{\bK,\overline{\tau}}$ where $\text{arg}C_0$ reprensents the angle of the complex number $C_0$.  Thus $\mathcal{F}(\bPsi_1,\bPsi_2)=C_{\omega_*}\begin{pmatrix}1\\\ii\\\end{pmatrix}$. Since $ \mathcal{E}_{\omega_*}$ is two-dimensional, $\mathcal{E}_{\omega_*}=\text{span}\{\bPsi_1\}\oplus\text{span}\{\bPsi_2\}$. In other words, $\mathcal{E}_{\omega_*}$ does not depend on the choices of $\bPsi_1$ and $\bPsi_2$.
\end{proof}

The above proposition states that any 2-dimensional eigenspace of $\MM_W$ at high symmetry point $\bK$ can be characterized by the eigenspaces of $\RRT$ on $\L^2_{\bK}$. We now turn to the behavior of the dispersion relation near $\bK$ at the degenerate eigenvalues. This plays an essential role in studying envelope dynamics. It turns out that the two-fold degeneracy implies a conical intersection of the dispersion relation which we shall give in the following theorem.

\begin{theo}\label{theo.diracpoint}
Let  $W(\bx)$ be a honeycomb material weight in  Definition \ref{honeycomb_material}. Assume $\omega_{_D}=\omega_b(\bK)=\omega_{b+1}(\bK)$ for some $b\ge 1$ is a two-fold degenerate eigenvalue of $\MM_W$ in $\L^2_\bK$ and $\mathcal{E}_{\omega_{_D}}\subset \L^2_{\bK,\tau}\oplus \L^2_{\bK,\overline{\tau}}$. Let
\begin{equation}\label{eq.cf}
C_D=\frac{1}{2}\Bigl | \mathcal{F}(\bPsi_1,\bPsi_2)\cdot \begin{pmatrix}1\\-\ii\\\end{pmatrix}\Bigr |,
\end{equation}
where $\bPsi_j,\, j=1,2$ are given in Proposition \ref{pro.psi}.
If  $C_D>0$, then
 $(\bK,\omega_{_D})$ is a Dirac point in the sense that there exist Lipschitz continuous functions $e_-(\bk)$, $e_+(\bk)$ and $\chi_0>0$, such that             \begin{equation*}
                \begin{array}{lll}
                \omega_{b+1}(\bk)-\omega_{_D} &=& +C_D|\bk-\bK|(1+e_+(\bk)),\\
                \omega_{b}(\bk)-\omega_{_D} &= &-C_D|\bk-\bK|(1+e_-(\bk)),
                \end{array}
            \end{equation*}
where  $|e_{\pm}(\bk)|<C|\bk-\bK|$ when $|\bk-\bK|<\chi_0$.
\label{theo.conical}
\end{theo}

\begin{proof}
To prove $(\bK,\omega_D)$ is a Dirac point, the key is to solve the eigenvalue problem at $\bk=\bK+\kappa$ with $\kappa=\begin{pmatrix} \kappa^{(1)}\\\kappa^{(2)}\end{pmatrix}$ being  sufficiently small,
\begin{equation}\label{eq:pert1}
    \begin{array}{l}
        \MM_W(\bK+\kappa)\mathbf{u}(\bx,\bK+\kappa) = \omega(\bK+\kappa)\mathbf{u}(\bx,\bK+\kappa),\\
        \bu(\bx+\vv,\bK+\kappa) = \bu(\bx,\bK+\kappa),
    \end{array}
    \end{equation}
    where $\MM_W(\bK+\kappa) = \MM_W(\bK)+W(\bx) B(\kappa)$, and $$B(\kappa)=\begin{pmatrix}
        0 & 0 &-\kappa^{(1)}\\
        0 & 0 &-\kappa^{(2)}\\
        -\kappa^{(1)}&-\kappa^{(2)} & 0\\
    \end{pmatrix}.$$

Regarding $W(\bx)B(\kappa)$ being a perturbation to the operator $\MM_\omega(\kappa)$, we follow the perturbation theory for the spectrum of linear operator.  We first expand the eigenvalue and the eigenfunctions as follows
\begin{equation}\label{eq.expan}
\omega(\bK+\kappa)=\omega_D+\tilde{\omega}, \quad \bu(\bx,\bK+\kappa)=\alpha_1\bu_1+\alpha_2\bu_2+\tilde{\bu},
\end{equation}
where $\bu_i=e^{-\ii\bK\cdot\bx}\bPsi_i$, $\la \bu_i, \tilde{\bu} \ra_W=0, ~ i=1,2$ and $\alpha_1,\alpha_2\in \mathbb{C}$ are to be determined.

Next, we substitute the expansions (\ref{eq.expan}) into  (\ref{eq:pert1}) and obtain
\begin{equation}
    (\MM(\bK)-\omega_D)\tilde{\bu} = (-W(\bx)B(\kappa)+\tilde{\omega})(\alpha_1\bu_1+\alpha_2\bu_2+\tilde{\bu}),
    \label{eq:mm}
\end{equation}
where we have used the fact  that $\text{Ker}\left(\MM(\bK)-\omega_D\right)=\text{span}\{\bu_1,\bu_2\}$.

Define the operators $P_{\parallel}$ and $P_{\bot}$  as $\forall \bu\in\L_\text{per}^2(\Lambda)$
\begin{equation*}
    P_{\parallel}\bu := \la\bu_1,\bu\ra_W\bu_1+\la\bu_2,\bu\ra_W\bu_2, \quad P_{\bot} = I - P_{\parallel}.
\end{equation*}
Operating $P_{\parallel}$ and $P_{\bot}$ on both sides of (\ref{eq:mm}) yields
\begin{align}
    (\MM(\bK)-\omega_D)\tilde{\bu} = P_\bot(-W(\bx)B(\kappa)+\tilde{\omega})(\alpha_1\bu_1+\alpha_2\bu_2+\tilde{\bu})\label{eq:mm1},\\
    0 = P_\parallel(-W(\bx)B(\kappa)+\tilde{\omega})(\alpha_1\bu_1+\alpha_2\bu_2+\tilde{\bu}).
    \label{eq:mm2}
\end{align}

For sufficiently small $\tilde{\omega}$, we solve $\tilde{\bu}$ in terms of $\bu_1$ and $\bu_2$ as
  \begin{equation}\label{eq:mm3}
      \tilde{\bu} = \hat{c}(\alpha_1\bu_1+\alpha_2\bu_2),
  \end{equation}
  where $\hat{c} = \Big(I-(\MM(\bK)-\omega_D)^{-1}P_\bot(-W(\bx)B(\kappa)+\tilde{\omega})\Big)^{-1}(\MM(\bK)-\omega_D)^{-1}P_\bot(-W(\bx)B(\kappa)+\tilde{\omega})$.

Plugging (\ref{eq:mm3}) into (\ref{eq:mm2}),  we obtain a linear system of algebraic equations for the undetermined coefficients $\alpha_1$ and $\alpha_2$ as
 \begin{equation}\label{eq.solv}
 G(\tilde{\omega},\kappa)\begin{pmatrix}\alpha_1\\\alpha_2\end{pmatrix}=0.
 \end{equation}
Here the $2\times 2$ matrix $G(\tilde{\omega},\kappa)=(G_{j,l})_{2\times2},~j,l=1,2$ is of the form
 \begin{equation*}
      G_{j,l} = \tilde{\omega}\delta_{jl}-\la\bPsi_j,W(\bx)B(\kappa)\bPsi_l\ra_{W}-O(|\kappa|(|\kappa|+|\tilde{\omega}|)),
 \end{equation*}
where we have used $\bu_j = e^{-\ii\bK\cdot\bx}\bPsi_j,~ j=1,2$.

Existence of nontrivial solutions of (\ref{eq.solv}) implies the solvability condition
    $$det(G(\tilde{\omega},\kappa)) = 0.$$
A direct calculation yields
\begin{equation}\label{eq.omegasq}
\tilde{\omega}^2-C_D^2|\kappa|^2+g_{21}(\tilde{\omega}, \kappa)+g_{12}(\tilde{\omega}, \kappa)+g_{03}(\tilde{\omega}, \kappa)=0,
\end{equation}
where $|g_{rs}(\tilde{\omega}, \kappa)| \leq C|\tilde{\omega}|^r|\kappa|^s$, $r,s\in\{0,1,2,3\}$ and we have used
$\la\bPsi_j,W(\bx)B(\kappa)\bPsi_l\ra_{W}=\kappa\cdot \mathcal{F}(\bPsi_j,\bPsi_l)$ by Proposition \ref{prop.f}.

Solving (\ref{eq.omegasq}) for sufficiently small $|\kappa|$ implies
$\tilde{\omega} = \pm C_D|\kappa|(1+O(|\kappa|))$,
which follows
 \begin{equation*}
\omega(\bK+\kappa) = \omega_{_D}+\tilde{\omega} = \omega_{_D}\pm C_D|\kappa|(1+e_\pm(\kappa)),
 \end{equation*}
 with $e_{\pm}(\kappa)=O(|\kappa|)$. This completes the proof.
\end{proof}

\subsection{Spectral gap opening under $\PPT\TTT$-symmetry breaking perturbations}\label{sec.gap}
In this subsection, we shall investigate the stability of the Dirac points under the symmetry-breaking perturbations.  Consider the perturbed material weight
\begin{equation*}
W^\delta(\bx) = W(\bx)+\delta V(\bx),
\end{equation*}
 with $\delta$ a small parameter. This work focuses on the $\PPT\TTT$-symmetry breaking perturbations. Specifically, we assume that the perturbed material weight $V(\bx)$ is a $\Lambda$-periodic $3\times 3$ Hermitian matrix and anti-$\PPT\TTT$-symmetry, i.e., $\overline{V(-\bx)}=-V(\bx)$. It immediately follows that $\PPT\TTT\MM_V=-\MM_V\PPT\TTT$ where $\MM_V=V(\bx)\LL$ with $\LL$  given in (\ref{eq.LL}).

To investigate the behavior of Dirac point $(\bK, \omega_{_D})$ under the above perturbation, we need to solve the perturbed eigenvalue problem
\begin{equation}\label{eigen_perturb}
\MM_{W^{\delta}} \bPsi^{\delta}:=(\MM_W+\delta\MM_V) \bPsi^{\delta}=\omega^\delta \bPsi^{\delta},\quad \bPsi^{\delta}\in \L^2_{\bK}.
\end{equation}

Substituting the asymptotic expansions
\begin{equation}\label{asy}
\bPsi^{\delta}=(\beta_1\bPsi_1+\beta_2\bPsi_2)+\widetilde{\bPsi},\quad \text{and }
\omega^\delta= \omega_{_D}+\tilde{\omega},
\end{equation}
to (\ref{eigen_perturb}) yields
\begin{equation*}
(\MM_W-\omega_{_D})\widetilde{\bPsi}=(\tilde{\omega}-\delta \MM_V)(\beta_1\bPsi_1+\beta_2\bPsi_2+\widetilde{\bPsi}),
\end{equation*}
where $\bPsi_1$ and $\bPsi_2$ are the eigenfunctions of $\MM_W$ corresponding to the Dirac point $(\bK,\omega_{_D})$ given in Theorem \ref{theo.diracpoint}, and $\beta_j\in\mathbb{C},~j=1,2$ are to be determined.

Similar to the proof of Theorem \ref{theo.diracpoint},  we shall obtain the solvability condition
$\det(\delta Q+U-\tilde{\omega})=0$,
where $U=(U_{ij})_{2\times2}$ with $U_{ij}=O(|\delta|(|\delta|+|\tilde{\omega}|))$ and $Q=\left(Q_{jl}\right)_{2\times2}, ~j,l=1,2$ is defined as
\begin{equation}\label{eq.bifurcation}
Q_{jl}= \la\bPsi_j,\MM_V \bPsi_l\ra_W.
\end{equation}

The key is to evaluate the entries of $Q$ explicitly. Note that $\MM_V$ is generally NOT a self-adjoint operator in $\L^2_\bK$ with the $W$-weighted inner product defined in \ref{eq.Wweight}. However, the matrix $Q$ is Hermitian. This claim can be directly obtained from the following calculation
\begin{equation*}
\begin{split}
    Q_{jl}&=\la\bPsi_j,\MM_V \bPsi_l\ra_W=\la\bPsi_j,V(\bx)W^{-1}(\bx)\MM_W \bPsi_l\ra_W\\
     &=\omega_{_D}\la\bPsi_j,V(\bx)W^{-1}(\bx) \bPsi_l\ra_W = \omega_{_D}\int_\Omega  \bPsi_j^*W^{-1}(\bx)V(\bx)W^{-1}(\bx) \bPsi_l\dd \bx,\\
\end{split}
\end{equation*}
where we have used the facts $\bPsi_j$, $j=1,2$ are the eigenfunctions of $\MM_W$ and $W^{-1}(\bx)V(\bx)W^{-1}(\bx)$ is a Hermitian matrix.
In other words, we have proved
\begin{equation}\label{eq.qrell}
 Q_{11},Q_{22}\in\mathbb{R}, \quad \text{and } Q_{12}=\overline{Q_{21}}.
\end{equation}

 Recalling $\bPsi_1=\PPT\TTT\bPsi_2$, we have by Proposition \ref{prop.inner}
\begin{equation}\label{eq.m1}
\begin{aligned}
&\la\bPsi_1,\MM_V \bPsi_1\ra_W=\la\pt\bPsi_2,\MM_V \pt\bPsi_2\ra_W\\
&=-\la\pt\bPsi_2,\pt\MM_V \bPsi_2\ra_W=-\overline{\la\bPsi_2,\MM_V \bPsi_2\ra_W}\\
\end{aligned}
\end{equation}
and
\begin{equation}\label{eq.m2}
\begin{aligned}
&\la\bPsi_1,\MM_V \bPsi_2\ra_W=\la\pt\bPsi_2,\MM_V \pt\bPsi_1\ra_W\\
&=-\la\pt\bPsi_2,\pt\MM_V \bPsi_1\ra_W=-\overline{\la\bPsi_2,\MM_V \bPsi_1\ra_W}.
\end{aligned}
\end{equation}

Summarizing the above calculations implies
\begin{equation}\label{eq.qs}
Q_{11}=-Q_{22}=\theta_\sharp\in\mathbb{R}, \quad Q_{12}=Q_{21}=0.
\end{equation}

With the simple form of $Q$ in (\ref{eq.qs}), we have
\begin{equation*}
\tilde{\omega}=\pm\delta\theta_\sharp+O(\delta^2).
\end{equation*}
As long as $\theta_\sharp\neq0$, the two-fold degenerate Dirac point $(\bK, \omega_{_D})$ splits into two simple eigenvalues: $\omega^\delta_\pm=\omega_{_D}\pm\delta\theta_\sharp+O(\delta^2)$. By the continuity of $\omega^\delta(\bk)$, we see that a local spectral gap opens under an anti-$\pt$-symmetric perturbation.  As a matter of fact, as long as the perturbation material weight $V(\bx)$ is NOT $\PPT\TTT$-symmetry, there is always a gap opening around the Dirac points with cumbersome calculations.

\subsection{Numerical examples of linear spectrum}
In this subsection, we show some simulations to demonstrate our analysis from the numerical aspect. Introduce a simple $\frac{2\pi}{3}$-rotation variant scalar function
\begin{equation}\label{eq.h}
h(\bx)=\cos(\bk_1\cdot\bx)+\cos(\bk_2\cdot\bx)+\cos(\bk_3\cdot\bx),
\end{equation}
where $\bk_1,\bk_2$ are the dual lattice vectors and $\bk_3=-\bk_1-\bk_2$.
 Obviously, $h(\bx)$ is $\Lambda$-periodic, even, and real, see Figure \ref{fig.mu3d}.
Based on $h(\bx)$, we construct the material weight $W(\bx)$ as follows
\begin{equation}\label{eq.parameter}
a(\bx)=(1-\frac{1}{5}h(\bx))^{-1}, \quad \text{and } A(\bx)=\begin{pmatrix}
    10-h(\bx) & 0\\
    0 & 10-h(\bx)\\
\end{pmatrix}.
\end{equation}
Apparently, $W(\bx)$ defined by (\ref{eq.parameter}) is a honeycomb material weight in the sense of Definition \ref{honeycomb_material}.

We solve the eigenvalue problem (\ref{eq:eigenk}) for the specific honeycomb material weight (\ref{eq.parameter}) for $\bk\in [-5,5] \times [-5,5]$ which contains the Brillouin zone $\mathcal{B}$.
 The numerical method which we use here is the Fourier collocation method, see \cite{Yang2010Nonlinear} for example.  The two dispersion surfaces $\omega_1(\bk)$ and $\omega_2(\bk)$, which are the smallest two positive eigenvalues of (\ref{eq:eigenk}), are shown in Figure \ref{fig.mu3d}. They conically intersect with each other at the vertices of the Brillouin zone.

\begin{figure}[htbp]
  \centering
  \includegraphics[width=5cm]{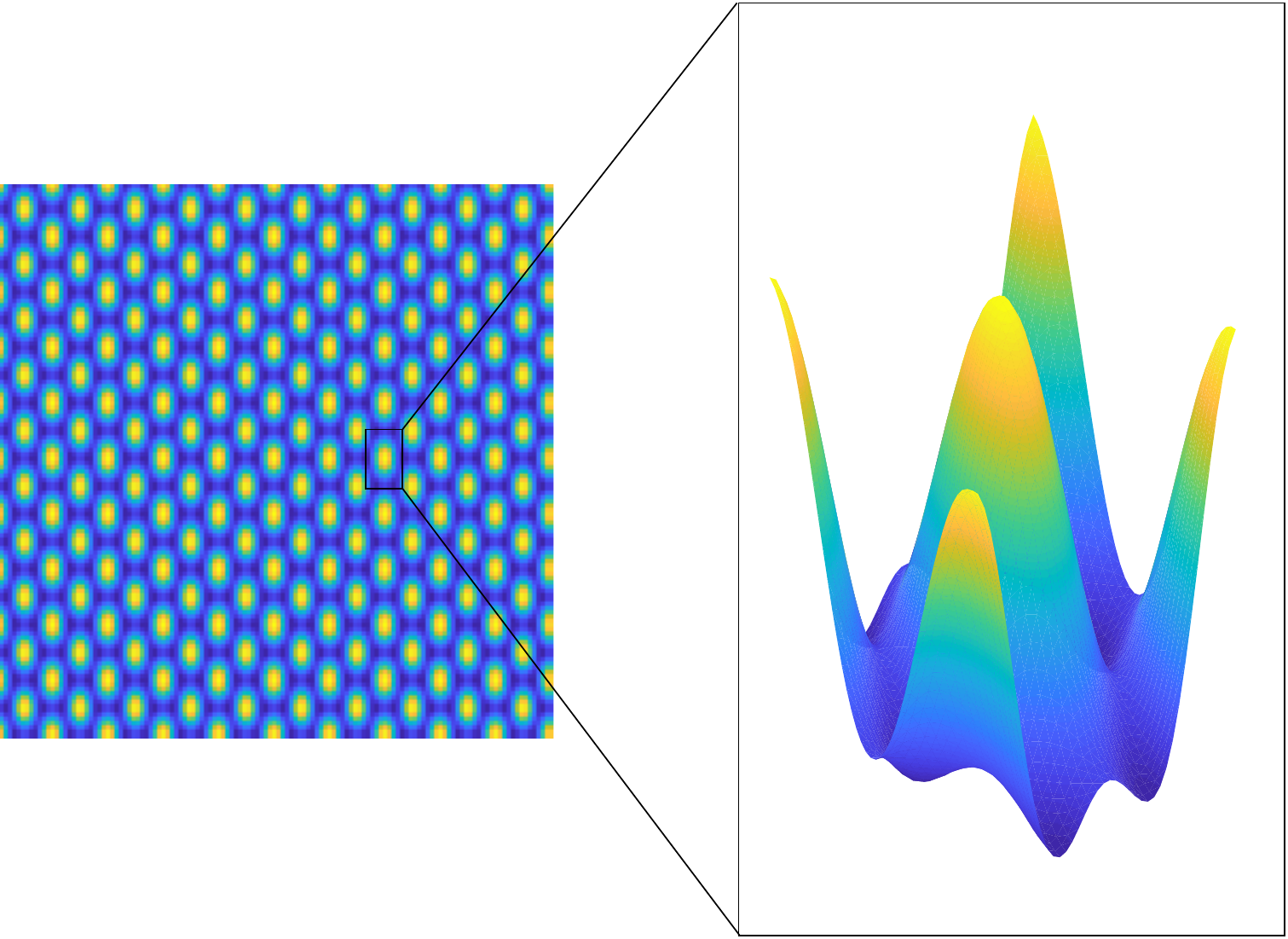}
  \includegraphics[width=6cm]{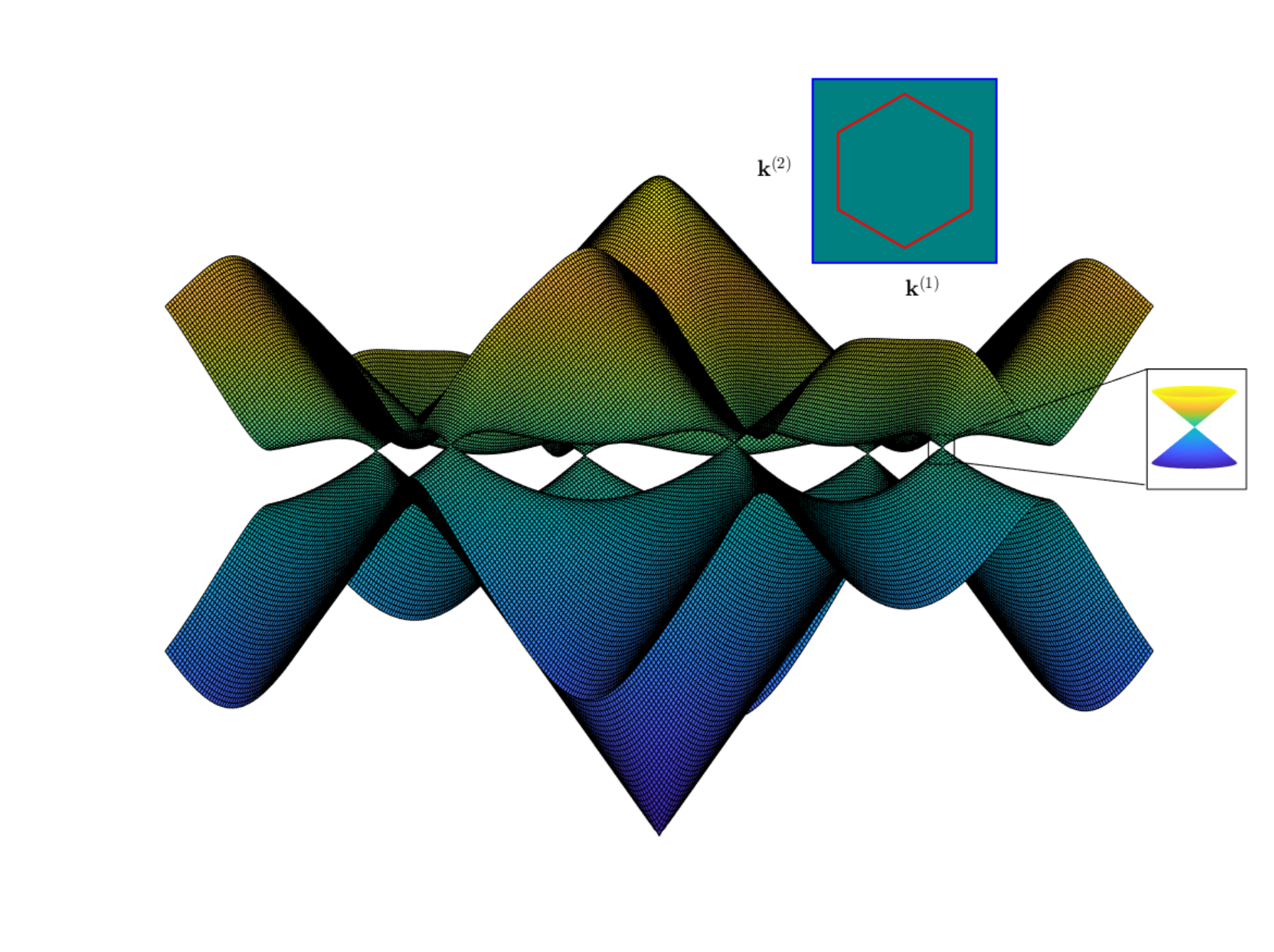}
  \caption{ Left panel: Image of $h(\bx)$ given in (\ref{eq.h}). Right panel: The lowest two positive dispersion surfaces $\omega_1(\bk)$ and $\omega_2(\bk)$ in the domain shown on the top. Dirac points occur at the intersections of two dispersion surfaces. The inset shows the zoomed-in Dirac cone at one Dirac point. }
  \label{fig.mu3d}
\end{figure}

To verify that the disappearance of the Dirac points under the $\PPT\TTT$-symmetry breaking perturbation, we numerically solve the perturbed eigenvalue problem (\ref{eigen_perturb}). We still use the honeycomb material weight $W(\bx)$ given in (\ref{eq.parameter}). The perturbation weight $V(\bx)$ is
\begin{equation}\label{eq.v}
V(\bx)= \begin{pmatrix}0 &\ii h(\bx) & 0\\
-\ii h(\bx)&0 &0\\
0&0 &0\\
 \end{pmatrix}.
 \end{equation}
Apparently, this perturbation weight  breaks the $\PPT\TTT$-symmetry as $\PPT\TTT\MM_V=-\MM_V\PPT\TTT$. Physically, the whole material weight $W^\delta$ corresponds to the magneto-optic material, see \cite{HR:07}. The results are shown in Figure \ref{fig.mu2d_break}. For clearness, we plot the dispersion relation along the direction of $\bk_2$ centered at $\bK$, i.e., $\bk = \bK + \lambda\bk_2$ with $\lambda\in [-0.5,0.5]$. The two branches of the dispersion relation disjoint with each other and a local gap appears once the perturbation is applied. Further we see that the gap increases in proportion to $\delta$. The numerical simulations agree well with our analysis.

\begin{figure}[htbp]
     \centering
     \begin{subfigure}[b]{0.4\textwidth}
         \centering
         \label{figh:a}\includegraphics[width=\textwidth]{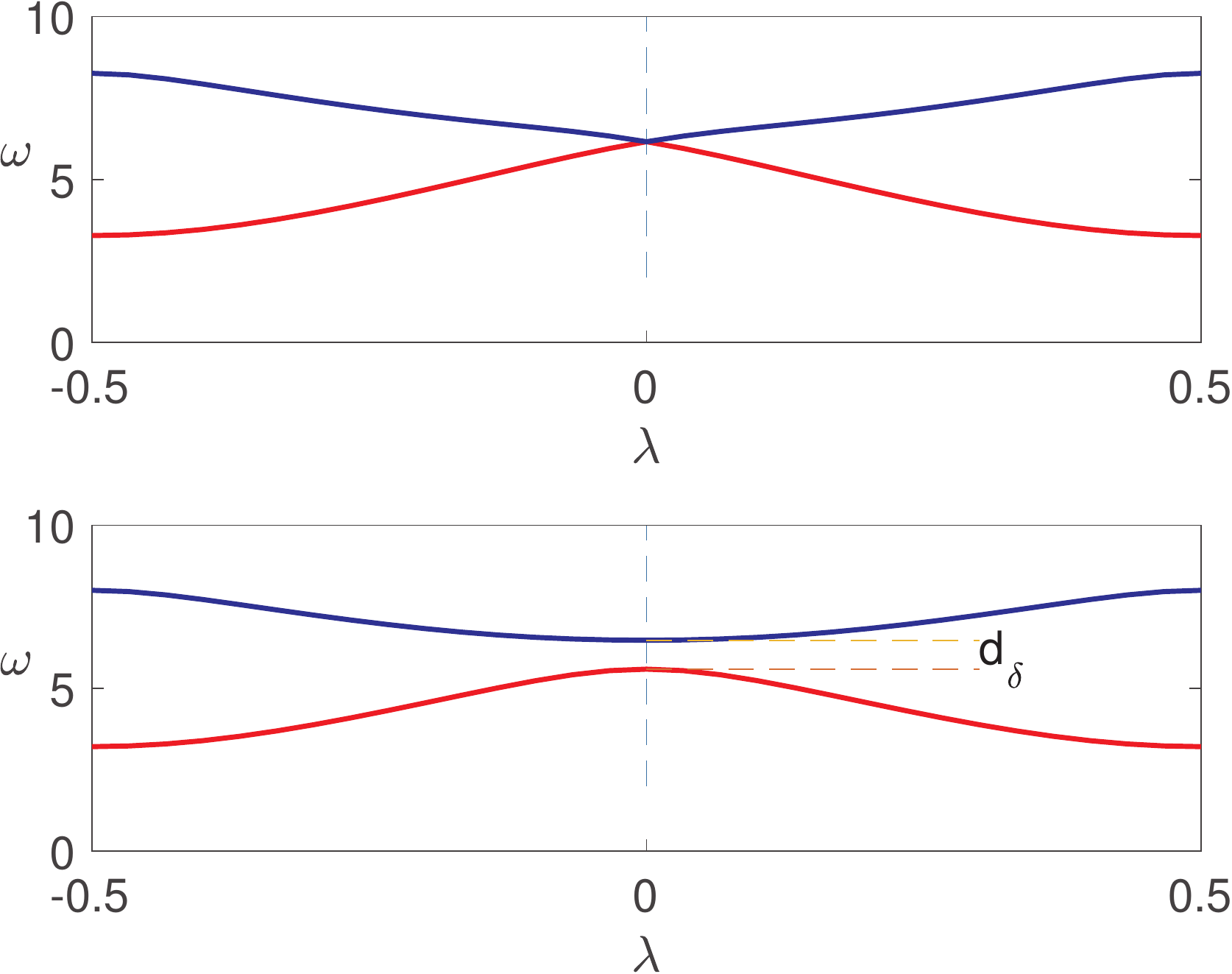}
         \end{subfigure}
     \hspace{2 em}
          \begin{subfigure}[b]{0.4\textwidth}
         \centering
         \label{figh:b}\includegraphics[width=\textwidth]{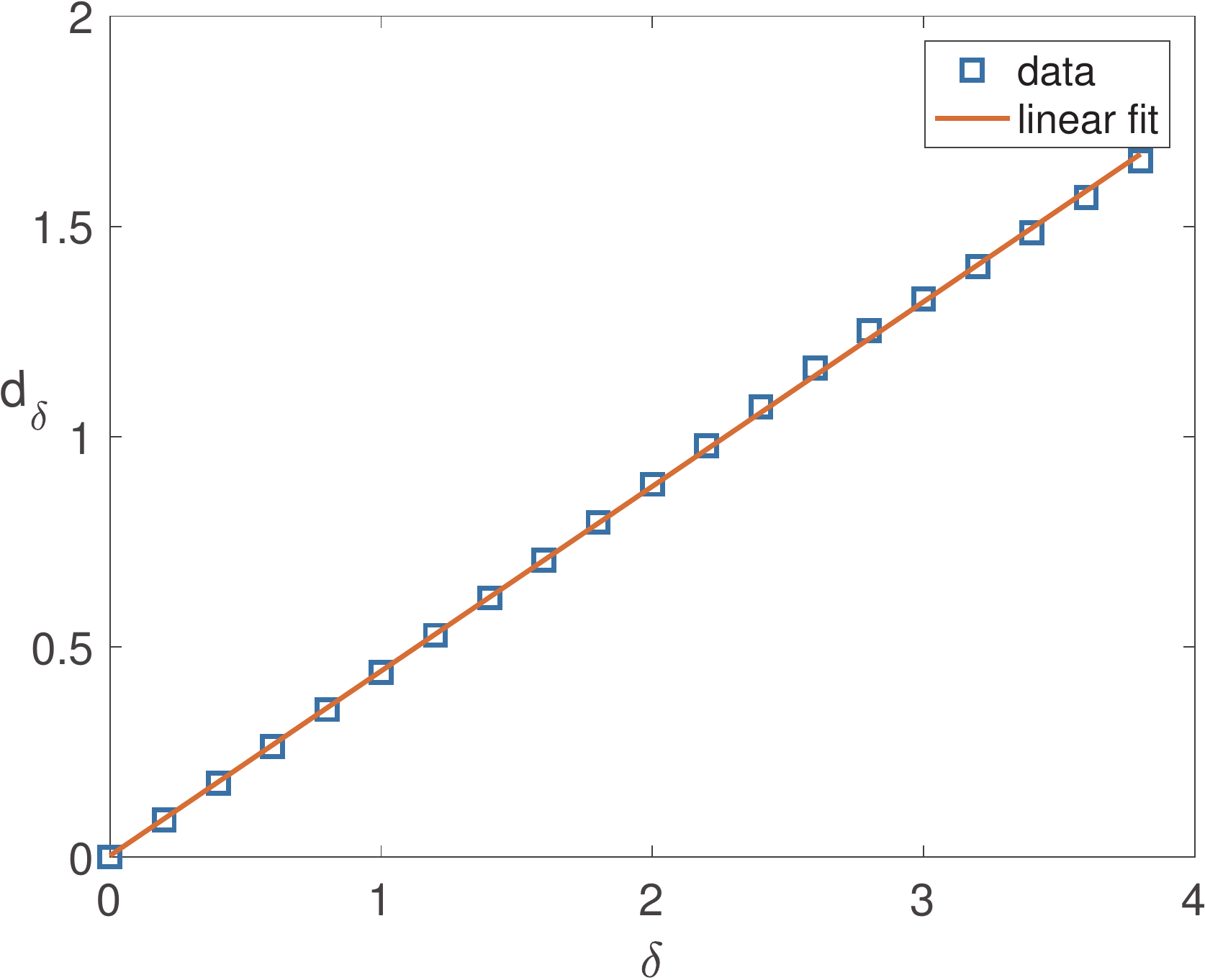}
     \end{subfigure}
       \caption{Left panel: the dispersion relation of the perturbed material weight $W^\delta(\bx)$ given in (\ref{eq.parameter}) and (\ref{eq.v}) along $\bk_2$ direction for $\delta=0$ (top) and $\delta=2$ (bottom). Right panel: the local gap width $d_\delta$ versus the perturbation parameter $\delta$.}
  \label{fig.mu2d_break}
\end{figure}

\section{Nonlinear dynamics of envelopes}\label{sec.envelope}
With the analytical structure of the linear spectrum near the Dirac points, we now can investigate the nonlinear dynamics of the envelope associated with the Dirac points. The derivation is presented in this section.

Turning back to the original Maxwell's equations (\ref{eq:maxwellmain})-(\ref{eq.constitutive}), we still divide the system of the equations into the TE component
\begin{equation}
    \ii\partial_t\bPsi_{e} + \MM_{W_e}\bPsi_{e} + \sigma N_e=0,
    \label{eq.TEnonlinear}
\end{equation}

 and the TM component
\begin{equation}
    \ii\partial_t\bPsi_{m} - \MM_{W_m}\bPsi_{m} + \sigma N_m=0,
    \label{eq.TMnonlinear}
\end{equation}
where $$N_e=\ii W_e(\bx)\begin{pmatrix}
        \partial_t (|\bE|^2\bPsi_e^\bot)\\
        0\\
        \end{pmatrix}
        , \quad N_m=\ii W_m(\bx)\begin{pmatrix}
        0\\
        \partial_t (|\bE|^2\bPsi_m^{(3)})\\
        \end{pmatrix}
        $$ are the nonlinear terms with $|\bE|^2=|\bPsi_e^\bot|^2+|\bPsi_m^{(3)}|^2$.
It is seen that the  TE and TM  components of electromagnetic waves are coupled to each other due to the nonlinearity. In this work, for simplicity, we will assume that one of the components is zero initially; for example, $\bPsi_m(\bx,0)=0$. Then, we only need to consider the TE component $\bPsi_{e}$ for $t>0$ since the TM component $\bPsi_{m}$ remains zero. The analysis of the other case, i.e., $\bPsi_e=0$  is the same.

We assume that the material weight is of the form
\begin{equation}
W_e(\bx)=W(\bx)+\delta \kappa(\delta\bx)V(\bx),
\label{eq.we}
\end{equation}
where  $W(\bx)$ is a honeycomb material weight in Definition \ref{honeycomb_material}, $V(\bx)$ is the perturbation stated in the previous section, $\delta\ll1$ is a small number and $\kappa(\delta\bx)$ is real and bounded.

\subsection{Derivation of the envelope equation}
Here we focus on the regime where the envelope scale, the modulation scale, and the nonlinearity effect are maximally balanced. Specifically, we assume that $\varrho:=\delta/\sigma=O(1)$.

Assume that the initial condition of (\ref{eq.TEnonlinear}) is $
\bPsi_e(\bx,0) = \sum_{j=1}^2\alpha_{j0}(\delta\bx)\bPsi_j(\bx)$,
where $\bPsi_j(\bx),~j=1,2$ are the eigenfunctions of $\MM_W$ corresponding the Dirac point $(\bK, \omega_{_D})$ given in Theorem \ref{theo.conical} and $\alpha_{j0}(\delta \bx), ~j=1,2$ are the slowly varying envelopes.

Introducing $\bX:= \delta \bx$, $T:= \delta t$ and $\bX=(X_1,X_2)^T$, we perform the standard multi-scale analysis. To this end, we first expand the solution $\bPsi_e(\bx,t)$ into
 the following asymptotic expansion
\begin{equation}\label{eq:ansatz}
    \bPsi_e(\bx,t) =  \sum_{j=1}^2\alpha_j(\bX,T)\bPsi_j(\bx)e^{\ii\omega_{_D}t}+\delta\bPsi_{e1}(\bx,t,\bX,T)e^{\ii\omega_{_D}t} +\cdots.
\end{equation}
Substituting (\ref{eq:ansatz}) into (\ref{eq.TEnonlinear}) yields a hierarchy of equations at different orders of $\delta$. The leading order is satisfied automatically.

At order $O(\delta)$,
\begin{equation}
\begin{array}{lll}
    (\MM_W-\omega_{_D})\bPsi_{e1} &= &-\ii(\bPsi_1\partial_T\alpha_1+\bPsi_2\partial_T\alpha_2)\\
    & &-(\alpha_1\kappa(\bX)\MM_V\bPsi_1+\alpha_2\kappa(\bX)\MM_V\bPsi_2)\\
    & &-(W(\bx)\LL_\bX\alpha_1\bPsi_1+W(\bx)\LL_\bX\alpha_2\bPsi_2)\\[0.05cm]
    & & +\omega_{_D}W(\bx)\begin{pmatrix}
        \|\alpha_1\bPsi_1^\bot+\alpha_2\bPsi_2^\bot\|^2(\alpha_1\bPsi_1^\bot+\alpha_2\bPsi_2^\bot)\\0\\
    \end{pmatrix} \\
&    :=& \Gamma_1+\Gamma_2+\Gamma_3+\Gamma_4,
\end{array}
\end{equation}
where $$\LL_\bX=\begin{pmatrix}
    0 & 0 &\ii\partial_{X_1}\\
    0 & 0 &\ii\partial_{X_2}\\
    \ii\partial_{X_1} & \ii\partial_{X_2} & 0 \\
\end{pmatrix}.$$

Applying the solvability conditions
\begin{equation}\label{eq:solvable}
    \la \bPsi_j,\sum_{k=1}^4\Gamma_k\ra_W=0,\, j=1,2
\end{equation}
yields the governing equations for the envelope dynamics.

The main task becomes the calculation of the solvability conditions (\ref{eq:solvable}). Now we  compute  (\ref{eq:solvable})  term by term.
First, the orthogonality between $\bPsi_1$ and $\bPsi_2$ yields
\begin{equation}\label{eq.g1}
\la\bPsi_1,\Gamma_1\ra_W = -\ii\partial_T\alpha_1, \quad \text{and } \la\bPsi_2,\Gamma_1\ra_W = -\ii\partial_T\alpha_2.
\end{equation}
By (\ref{eq.qs}), we obtain
\begin{equation}\label{eq.g2}
\la\bPsi_1,\Gamma_2\ra_W = -\theta_\sharp\kappa(\bX)\alpha_1, \quad \text{and } \la\bPsi_2,\Gamma_2\ra_W = \theta_\sharp\kappa(\bX)\alpha_2.
\end{equation}
Recall that
\[
\mathcal{F}(\bPsi_1,\bPsi_2)=C_D\begin{pmatrix}
            1\\
            \ii\\
        \end{pmatrix}
        ,\,\,
\mathcal{F}(\bPsi_2,\bPsi_1)=C_D\begin{pmatrix}
            1\\
            -\ii\\
        \end{pmatrix}\]
and $\mathcal{F}(\bPsi_1,\bPsi_1)=\mathcal{F}(\bPsi_2,\bPsi_2)=0.$
A direct calculation yields that

\begin{equation*}
\la\bPsi_j,W(\bx)\LL_\bX\alpha_k\bPsi_k\ra_W= \ii\nabla_\bX\alpha_k\cdot\mathcal{F}(\bPsi_j,\bPsi_k).\\
\end{equation*}
It follows that
\begin{equation}\label{eq.g3}
\begin{aligned}
\la\bPsi_1,\Gamma_3\ra_W = -\ii C_D\nabla_\bX\alpha_2\cdot\begin{pmatrix}1\\\ii\\ \end{pmatrix},\\
\la\bPsi_2,\Gamma_3\ra_W = -\ii C_D\nabla_\bX\alpha_1\cdot\begin{pmatrix}1\\-\ii\\ \end{pmatrix}.
\end{aligned}
\end{equation}

Let  $\tilde{\bPsi}_j=\bigl(\begin{smallmatrix}
    \bPsi_j^\bot\\
    0\\
\end{smallmatrix}\bigr)$, $j=1,2$. Obviously,  we have  $\tilde{\bPsi}_j\in\L^2_{\bK,\sigma_j}$ since $\bPsi_j\in\L^2_{\bK,\sigma_j}$, $\sigma_1=\tau$ and $\sigma_2=\overline{\tau}$.
Applying the operator $\RRT$, we have
\begin{equation*}
    \begin{array}{lll}
     \la\tilde{\bPsi}_j,W(\bx)(\tilde{\bPsi}_n^*\tilde{\bPsi}_l)\tilde{\bPsi}_k\ra_W &=&\la\RRT\tilde{\bPsi}_j,\RRT W(\bx)(\tilde{\bPsi}_n^*\tilde{\bPsi}_l)\tilde{\bPsi}_k\ra_W \\
    &=&\la\RRT\tilde{\bPsi}_j, W(\bx)((\RRT\tilde{\bPsi}_n)^*\RRT\tilde{\bPsi}_l)\RRT\tilde{\bPsi}_k\ra_W \\
    &=&\overline{\sigma}_j\overline{\sigma}_n\sigma_l\sigma_k\la\tilde{\bPsi}_j, W(\bx)(\tilde{\bPsi}_n^*\tilde{\bPsi}_l)\tilde{\bPsi}_k\ra_W. \\
        \end{array}
\end{equation*}
Note that the above term vanishes if $\overline{\sigma}_j\overline{\sigma}_n{\sigma}_l{\sigma}_k\neq1$. A simple enumeration implies there exist $\beta_1,\beta_2\in\mathbb{R}$ such that

\[    \la\tilde{\bPsi}_j,W(\bx)(\tilde{\bPsi}_n^*\tilde{\bPsi}_l)\tilde{\bPsi}_k\ra_W    = \begin{cases}
        \beta_1  & \text{if}~  (j,k,n,l)=(1,1,1,1),(2,2,2,2)\\
        \beta_2 & \text{if}~
            (j,k,n,l)=(2,1,1,2),(1,1,2,2),(2,2,1,1),(1,2,2,1)\\
        0 & \text{if}~ \text{otherwise}
    \end{cases}.
\]

Therefore, we obtain
\begin{equation}\label{eq.g4}
\begin{array}{l}
\la\Psi_1,\Gamma_4\ra_W = \omega_{_D}\varrho(\beta_1|\alpha_1|^2\alpha_1+2\beta_2|\alpha_2|^2\alpha_1),\\
\la\Psi_2,\Gamma_4\ra_W =\omega_{_D}\varrho(\beta_1|\alpha_2|^2\alpha_2+2\beta_2|\alpha_1|^2\alpha_2).
\end{array}
\end{equation}

Finally, the governing equation of the envelope associated with Dirac point is obtained by collecting (\ref{eq:solvable})-(\ref{eq.g4}). It is a nonlinear Dirac equation with varying mass which reads
\begin{equation}\label{eq:Dirac}
\begin{split}
\ii\alpha_{1_T}+\theta_\sharp\kappa\alpha_1 +\ii C_D\nabla_\bX\alpha_2\cdot \begin{pmatrix}
    1\\
    \ii\\
\end{pmatrix}-\omega_{_D}\varrho(\beta_1|\alpha_1|^2\alpha_1+2\beta_2|\alpha_2|^2\alpha_1)=0\\
\ii\alpha_{2_T}-\theta_\sharp\kappa\alpha_2 +\ii C_D\nabla_\bX\alpha_1\cdot \begin{pmatrix}
    1\\
    -\ii\\
\end{pmatrix}-\omega_{_D}\varrho(\beta_1|\alpha_2|^2\alpha_2+2\beta_2|\alpha_1|^2\alpha_2)=0
\end{split} \quad .
\end{equation}

For simplicity, we define $\tilde{\bX} = 1/C_D\bX$, $\tilde{\kappa}=\theta_\sharp\kappa$,   $p_1=-\omega_{_D}\varrho\beta_1$, and $p_2=-2\omega_{_D}\varrho\beta_2$. Drop the tilde notation, and we cast (\ref{eq:Dirac}) into
\begin{equation}
    \ii\partial_T\boldsymbol{\alpha} + (\ii\sigma_1\partial_{X_1}-\ii\sigma_2\partial_{X_2}+\kappa(\bX)\sigma_3 )\boldsymbol{\alpha}+\gamma(|\alpha_1|,|\alpha_2|)\boldsymbol{\alpha}=0,
    \label{eq.Dirac2}
\end{equation}
where $\boldsymbol{\alpha}=(\alpha_1,\alpha_2)^T=\bigl(\alpha_1(\bX,T),\alpha_2(\bX,T)\bigr)^T$,
$$\gamma(|\alpha_1|,|\alpha_2|)=\begin{pmatrix}[l]
        p_1|\alpha_1|^2+p_2|\alpha_2|^2 & 0\\
        0&p_1|\alpha_2|^2+p_2|\alpha_1|^2\\
    \end{pmatrix}$$ represents the  nonlinear effect, and $\sigma_1,\sigma_2,\sigma_3$ are Pauli matrices defined as
\begin{equation*}
    \sigma_1 = \begin{pmatrix}
        0&1\\
        1&0\\
    \end{pmatrix},\,
    \sigma_2 = \begin{pmatrix}
        0&-\ii\\
        \ii&0\\
    \end{pmatrix},\,
    \text{and}\,
    \sigma_3 = \begin{pmatrix}
        1& 0\\
        0&-1\\
    \end{pmatrix}.
\end{equation*}

\subsection{A typical numerical comparison}
Our derivation is based on a formal multi-scale analysis. In this subsection, we numerically justify the derivation via a typical comparison between the original Maxwell's equations (\ref{eq.TEnonlinear}) and the reduced envelope equation (\ref{eq:Dirac}). We shall simulate the typical topologically protected wave motion which possesses the chirality and immunity. We use the following physical setups. $W_e(\bx)=W(\bx)+\delta \kappa(\delta\bx)V(\bx)$ with $W(\bx)$ and $V(\bx)$ given in (\ref{eq.parameter}) and (\ref{eq.v}) and $\delta=0.05$. The modulation $\kappa(\delta\bx)$ plays an essential role in these topological phenomena. We choose  a smooth domain wall function $\kappa(\delta\bx)=\tanh(\delta(x_1-f(x_2))$ where the curve $x_1=f(x_2)$ defines the ``edge" between two materials. In the simulation, the edge curve presented by the white curve shown in Figure \ref{fig.compare} is composed of some end-to-end straight lines. With a well prepared initial condition, it is expected to see the wave propagate along the edge curve unidirectionally. It will be seen in the next section that this pattern can not persist in a strong nonlinear medium. Thus we ignore the nonlinearity in the numerical comparison.

We first use the Fourier collocation method, see for example \cite{Yang2010Nonlinear}, to obtain the normalized eigenfunctions $\bPsi_1$ and $\bPsi_2$ of the operator $\MM_{W}$ at Dirac point $\bK=\frac{1}{3}(\bk_1-\bk_2)$. Then, Maxwell's equations (\ref{eq.TEnonlinear}) is numerically solved with the following initial input
\begin{equation}\label{eq.maxwelli}
\bPsi_e(\bx,0) = \beta_1(\delta\bx)\bPsi_1(\bx) +\beta_2(\delta\bx)\bPsi_2(\bx),
\end{equation}
where $\beta_j(\delta\bx),~j=1,2$ denote the initial envelope which will match the initial condition for the envelope equation.

To do the comparison, we simulate the envelope equation (\ref{eq:Dirac}). The parameters $C_D=1.76$ and $\theta_\sharp=0.51$ are calculated numerically via the formulae (\ref{eq.cf}) and (\ref{eq.bifurcation}) with the same $\bPsi_1$ and $\bPsi_2$ above. The initial conditions are also consistent, i.e.,
\begin{equation}\label{eq.travelc}
\begin{pmatrix}
\alpha_1(\bX,0)\\
\alpha_2(\bX,0)\\
\end{pmatrix}=\begin{pmatrix}
\beta_1(\bX)\\ \beta_2(\bX)
\end{pmatrix}.
\end{equation}

To capture the topologically protected wave propagation, we use the following initial envelope
\begin{equation}\label{eq.envini}
\begin{pmatrix}
\beta_1(\bX)\\
\beta_2(\bX)\\
\end{pmatrix}=\begin{pmatrix}
1\\
-\ii\\
\end{pmatrix}\sech^{\frac{\theta_\sharp}{C_D}}(X_1-X_{10})e^{-0.2( X_2-X_{20})^2},
\end{equation}
where the initial central position $(X_{10},X_{20})$ is on the edge. It is noted that this initial envelope corresponds to the topologically protected linear edge state, see next section for details.

In both simulations, the pseudo-spectral method is used for spatial derivatives, and fourth-ordered Runge-Kutta method is used for time integration. The results are shown in Figure \ref{fig.compare}. It is seen in both simulations that the waves propagate along the edge with little energy leaking into the bulks. Moreover, the simulation for the envelope equation can perfectly capture the wave profile and its position under propagation.

The numerical comparison well justifies our derivation of the envelope equation.  We want to point out the original Maxwell's equations have highly oscillatory periodic structure while the envelope equation homogenizes the periodic structures and only describes the behaviors of the envelopes. Thus the envelope equation is a much simpler equation to study the interesting topological phenomena both analytically and numerically. In our simulations, in order to resolve the periodic structures, we need to choose very small space and time steps. This can not be easily implemented on a desktop computer. We run the simulation in a high-performance GPU server with a GPU of Tesla K40c and the computation cost about 10.9 hours. On the other hand, the numerical computation of the envelope equation (\ref{eq:Dirac}) only took about 12.3 seconds on the desktop computer with Xeon(R) CPU E5-26900 @2.90GHz.

\begin{figure}[htbp]
   \centering
   \begin{subfigure}[b]{\textwidth}
     \centering
     \includegraphics[width=\textwidth]{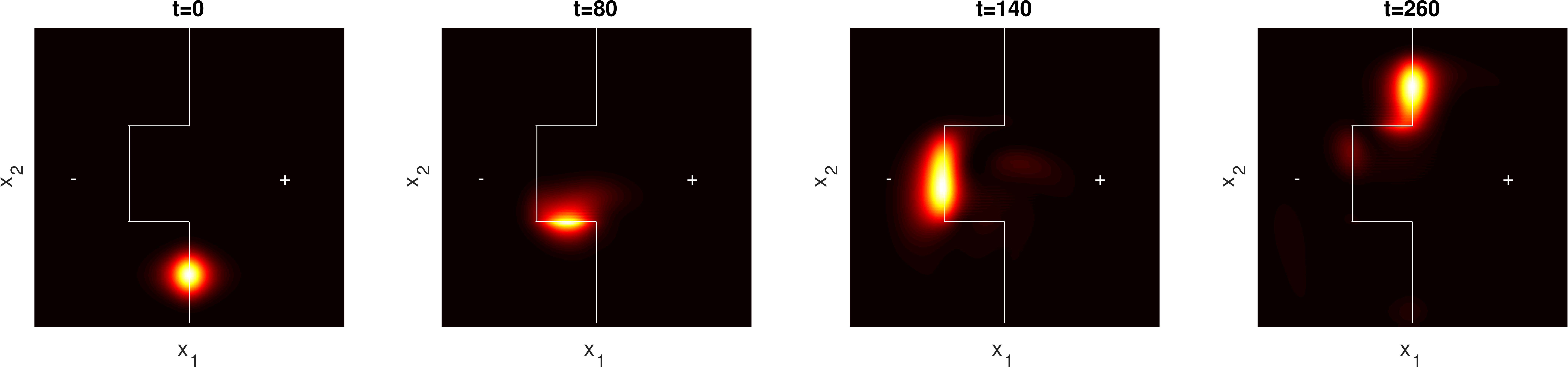}
     \caption{}
     \label{fig.maxwell}
     \end{subfigure}\\
   \begin{subfigure}[b]{\textwidth}
     \centering
     \includegraphics[width=\textwidth]{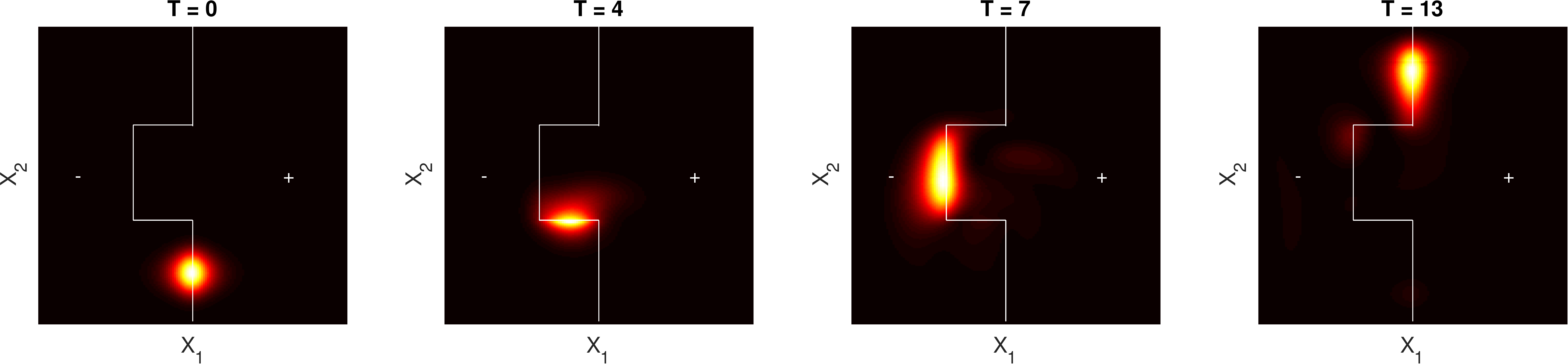}
     \caption{}
     \label{fig.envelope}
     \end{subfigure}
    \caption{The numerical simulations of Maxwell's equations (\ref{eq.TEnonlinear})  and the reduced envelope equation (\ref{eq:Dirac}). The snapshots of energies are shown at four successive time. Top panel (a): Energy $E=\bPsi^*W^{-1}\bPsi$ of Maxwell's equations.
 Bottom panel (b): Corresponding envelope energy $E = |\alpha_1|^2+|\alpha_2|^2$. The $\pm$ represent the signs of the domain wall function $\kappa(\cdot)$ in the areas.}
      \label{fig.compare}
\end{figure}

\section{Analysis of the envelope equation}\label{sec.dynamics}
In this section, we demonstrate that the envelope equation derived in the last section can describe many interesting topological wave propagations. To this end, we require that the ``mass" $\kappa(\bX)$ should change sign when passing through a given curve on which the mass vanishes. This corresponds to the physical setup in which two topologically different materials are glued together along with this cure, which is referred to as the edge. For simplicity, we choose $\kappa(\bX)=\tanh(f(X_1)-X_2)$ where $f(X_1)$ is a given continuous function. In this scenario, the edge is the curve defined by $X_2=f(X_1)$. The $\kappa(\bX)$ is negative above the curve and is positive below the curve. In the rest of this section, we use a lot of numerical simulations to show our results. For the time evolution simulations, the pseudo-spectral method with fourth order Runge-Kutta time integration is used \cite{Yang2010Nonlinear,bao2017numerical}. We use periodic boundary conditions and a huge computing domain such that the boundary effects do not pollute the fields in the center. To seek nonlinear modes, we adopt the Newton-conjugate-gradient method \cite{Yang2009Newton}. We have enlarged the computing domain and refined the mesh size to confirm the reliability of our numerical schemes.

\subsection{Linear and nonlinear line modes}
We first investigate the case where the edge is a straight line, i.e., $f(X_1)\equiv 0$. Without nonlinearity, (\ref{eq.Dirac2}) admits travelling line modes of the form
\begin{equation*}
\begin{pmatrix}
\alpha_1(\bX,T)\\
\alpha_2(\bX,T)\\
\end{pmatrix}=\begin{pmatrix}
1\\
1\\
\end{pmatrix}\sech(X_2)e^{i\xi(X_1-T)},
\end{equation*}
where $\xi\in \mathbb{R}$ is the line mode wave number. It is seen that line modes are product of exponentially decaying function in $X_2$ and plane waves in $X_1$. That is, line modes are localized at the edge and propagating along the edge. These line modes are referred to as topologically protected edge states. Interestingly, the linear modes with different wave numbers $\xi$ have the same velocity. Moreover, this envelope equation with the fixed $\kappa(\cdot)$ does not support the line modes which move in the opposite direction. This is related to the chirality of wave propagation in topological materials. We also note that the line modes in this system are not dispersive. It immediately follows that for any $g(X_1)\in L^2(\mathbb{R})$,
\begin{equation}\label{eq.travel}
\begin{pmatrix}
\alpha_1(\bX,T)\\
\alpha_2(\bX,T)\\
\end{pmatrix}=\begin{pmatrix}
1\\
1\\
\end{pmatrix}\sech(X_2)g(X_1-T),
\end{equation}
is an exact solution to (\ref{eq.Dirac2}) without nonlinearity.

The solution (\ref{eq.travel}) reveals the linear equation supports fully localized traveling wave solutions. A numerical example is shown in Figure \ref{fig.edge1}.  Moreover, when the straight-line-edge becomes a curved edge, these solutions travel along the edge with very little energy leaking to the bulk. In Figure \ref{fig.halfcircle}, we show a typical propagating pattern where the edge is a half-circle connected by two straight lines. This robust wave propagation pattern is related to the so-called topologically protected wave propagation and topological insulators \cite{cheng2016robust, HR:07}. Actually, the reduced envelope equations can describe many other complicated propagation patterns. Due to the length and scope of this paper, we leave these studies to future works.

\begin{figure}[htbp]
     \centering
     \begin{subfigure}[b]{0.3\textwidth}
         \centering
         \label{fig.kappa}\includegraphics[width=\textwidth]{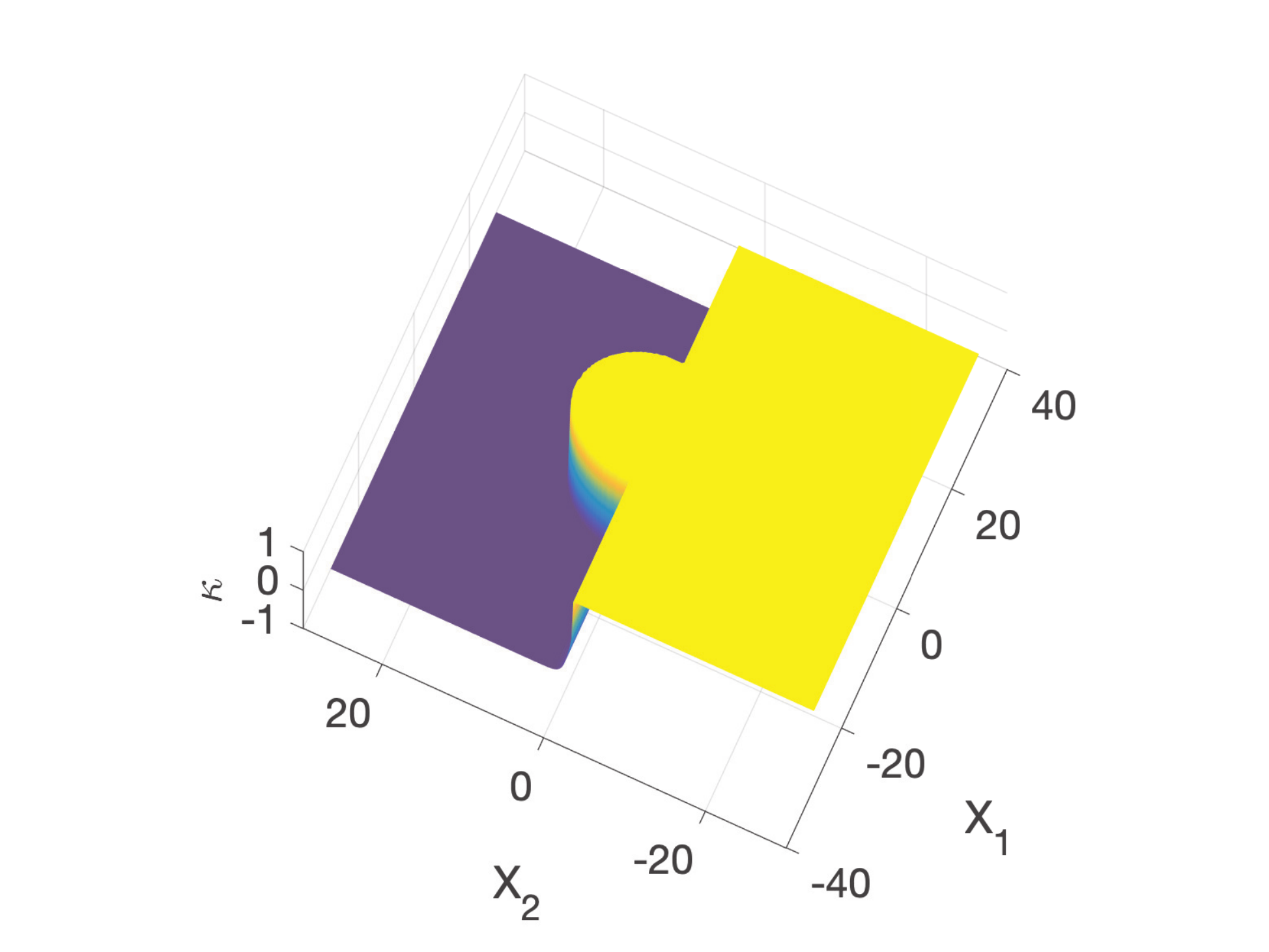}
         \caption{$\kappa(\bX)$}
         \end{subfigure}
     \hspace{1 em}
          \begin{subfigure}[b]{0.6\textwidth}
         \centering
         \label{figh:b}\includegraphics[width=\textwidth]{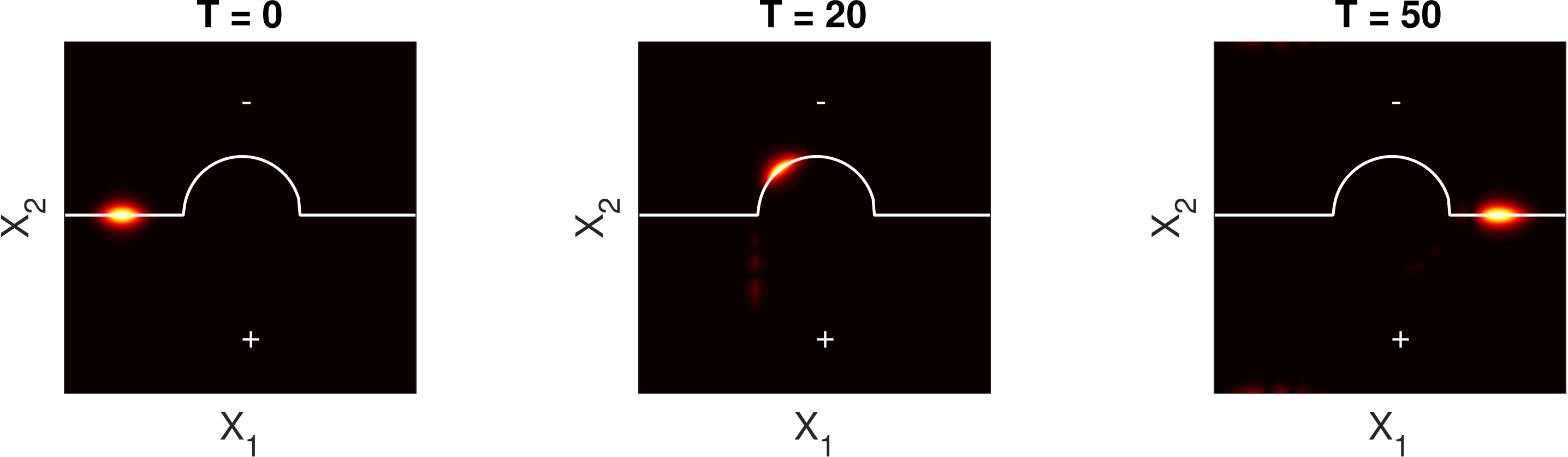}
         \caption{Evolution of modulated linear line mode in the linear media}
     \end{subfigure}
        \caption{Left panel (a): The surface plot of $\kappa(\bX)$. Right panel (b): The snapshots of $|\alpha_1|$ in the Dirac equation (\ref{eq.Dirac2}) without nonlinearity at three successive time.  The white curve denotes the edge.}
    \label{fig.halfcircle}
\end{figure}

In optics, the nonlinear effects cannot be neglected if the intensity of the electromagnetic waves propagating in the material is strong. Thus, it is important to investigate how the nonlinearity affects the interesting linear propagation patterns shown above. In Figure \ref{fig.linemodee1}, we present the propagation of the linear line mode Figure \ref{fig.linemode1}  in the nonlinear media where nonlinear parameters are $p_1=2$ and $p_2=1$. Here the edge is the straight-line $X_2=0$. We see that the line mode is destroyed by the nonlinearity, and a large portion of the energy is leaking to the bulk.


\begin{figure}[htbp]
     \centering
     \begin{subfigure}[b]{0.2\textwidth}
         \centering
         \includegraphics[width=\textwidth]{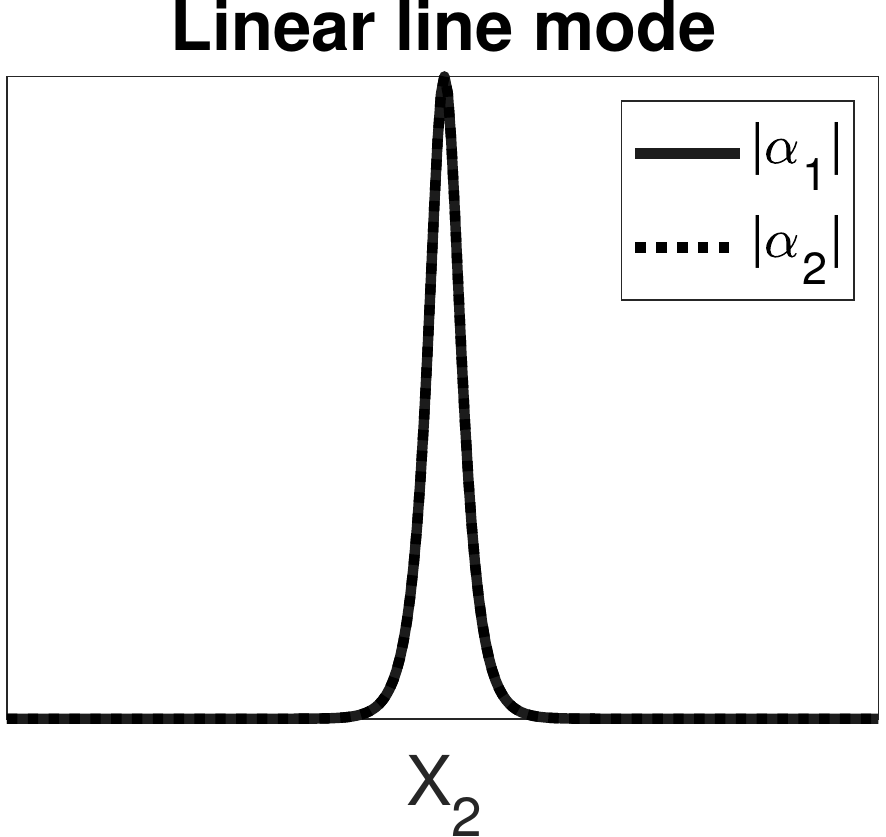}
         \caption{}
         \label{fig.linemode1}
         \end{subfigure}
         \hspace{1em}
     \begin{subfigure}[b]{0.7\textwidth}
         \centering
         \includegraphics[width=\textwidth]{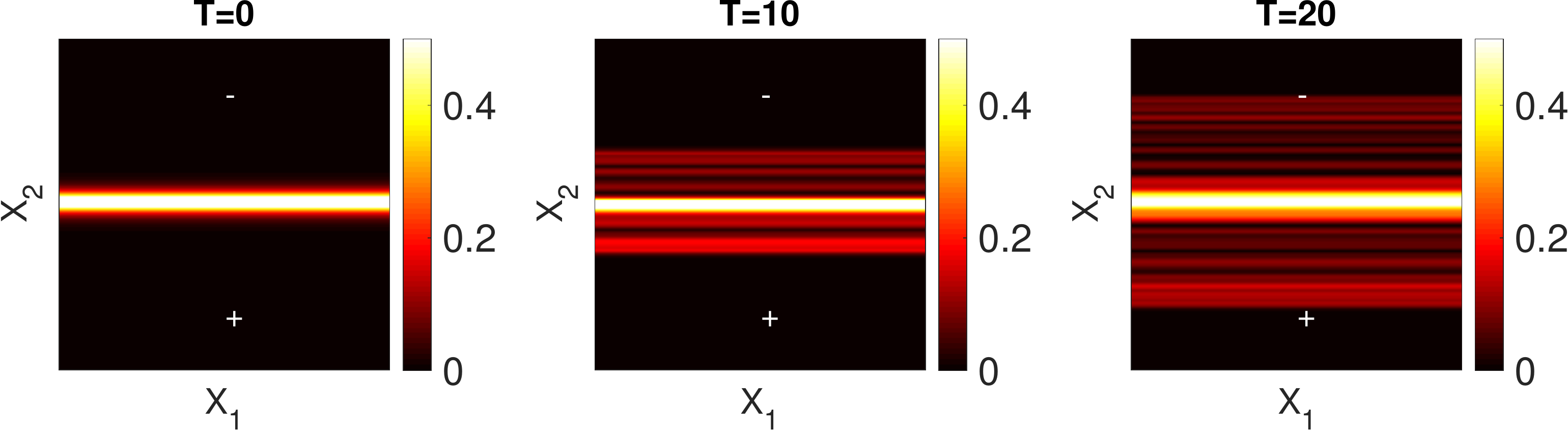}
         \caption{Evolution of linear line mode  in nonlinear media}
         \label{fig.linemodee1}
         \end{subfigure}
         \\
      \begin{subfigure}[b]{0.2\textwidth}
         \centering
         \includegraphics[width=\textwidth]{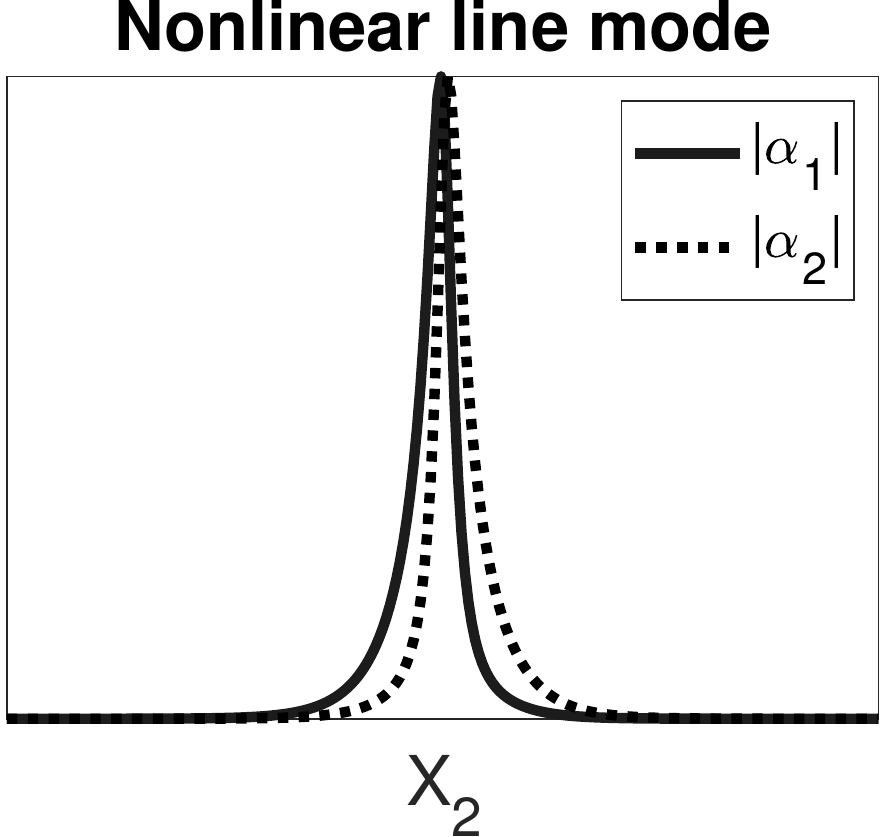}
         \caption{}
         \label{fig.linemode2}
         \end{subfigure}
                  \hspace{1em}
     \begin{subfigure}[b]{0.7\textwidth}
         \centering
         \includegraphics[width=\textwidth]{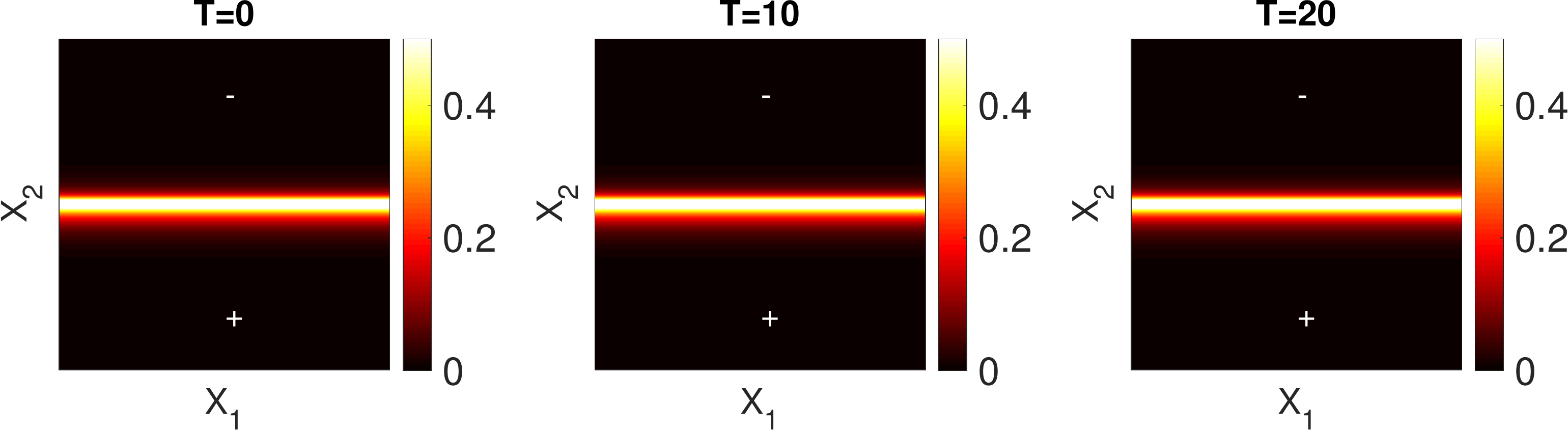}
         \caption{Evolution of nonlinear line mode  in nonlinear media}
         \label{fig.linemodee2}
         \end{subfigure}
        \caption{ Left panels (a,c): The cross section of  $|\alpha_1|$ and $|\alpha_2|$ of the linear and nonlinear line modes.  Right panels (b,d): The propagation of linear and nonlinear line modes in nonlinear media corresponding to $\kappa(\bX)=\tanh(-X_2)$ at three successive time.}
\end{figure}

Since the line modes could not survive in a nonlinear media, an interesting question to ask is whether there exist nonlinear line modes. To this end, we seek the solution to the nonlinear equation (\ref{eq.Dirac2}) of the form
\begin{equation}\label{eq.heigen}
    \begin{pmatrix}
        \alpha_1(\bX,T)\\
        \alpha_2(\bX,T)\\
    \end{pmatrix} =
    e^{-\ii\mu T}\begin{pmatrix}
        \chi_1(\bX)\\
        \chi_2(\bX)\\
    \end{pmatrix},
\end{equation}
where $\mu$ is the propagation constant. With the Newton-conjugate-gradient method \cite{Yang2009Newton}, we indeed find nonlinear line modes of which the profiles are shown in Figure \ref{fig.linemode2}. Compared to the linear line modes, which are symmetric, both components of the nonlinear line modes are asymmetric. We also want to point out that the nonlinear line modes that we obtain are not moving.

\begin{figure}[htbp]
     \centering
     \begin{subfigure}[b]{\textwidth}
         \centering
         \includegraphics[width=\textwidth]{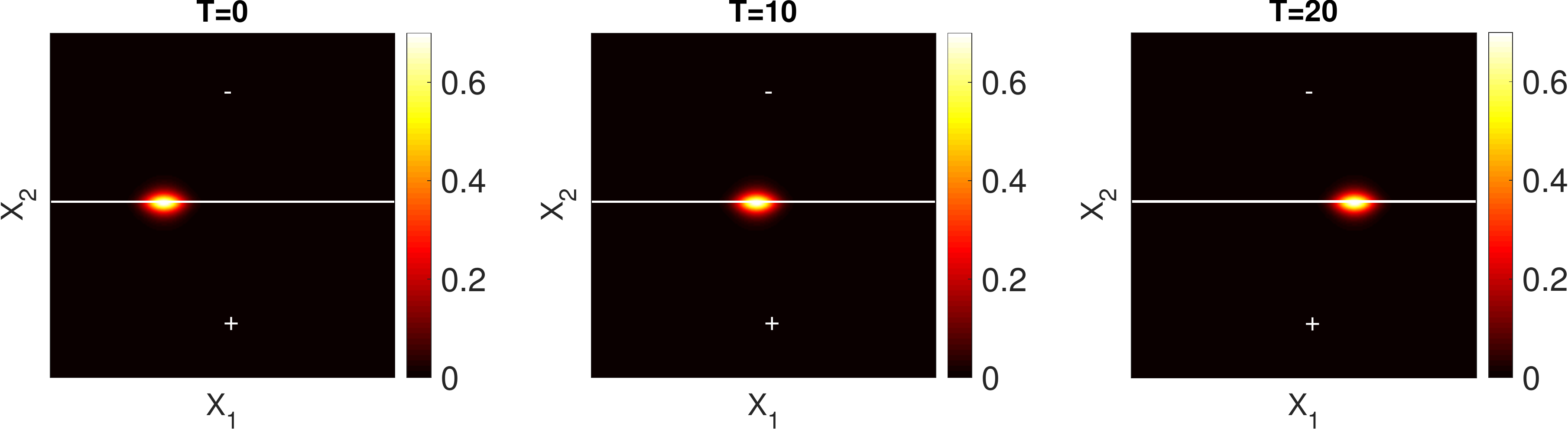}
         \caption{}
         \label{fig.edge1}
         \end{subfigure}\\
     \begin{subfigure}[b]{\textwidth}
         \centering
         \includegraphics[width=\textwidth]{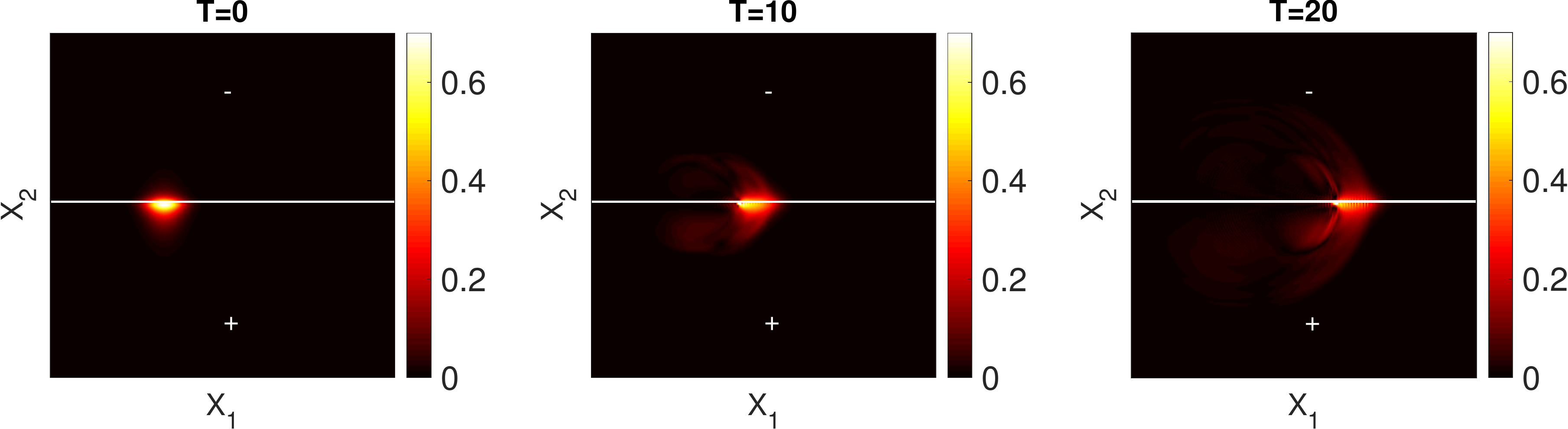}
          \caption{}
         \label{fig.edge2}
         \end{subfigure}
        \caption{The propagations of line modes modulated by a Gaussian in $X_1$ direction. The intensities $|\alpha_1|$ are shown at different time.
  Top panel (a): Linear modes in the linear media. By (\ref{eq.travel}), the initial input travels along the edge retaining its shape.
  Bottom panel (b): Nonlinear modes in the nonlinear media. A large portion of the energy leaks to the bulk due to the modulation.
  }
\end{figure}

\subsection{Fully localized nonlinear modes}
In the last subsection, we show that (\ref{eq.Dirac2}) admits nonlinear line modes. Line modes are localized along one direction, which means they have infinite energy by considering them as two-dimensional wave modes.  In real applications, it is interesting to investigate the propagation of waves with finite energy. In Figure \ref{fig.edge2}, we show the propagation of a nonlinear line mode modulated by a Gaussian along $X_1$ direction in the nonlinear media. It is seen that there exists considerable energy leaking to the bulk under propagation. This inspires us to seek fully localized nonlinear modes, i.e., solitary waves in this nonlinear system. By numerical iterations \cite{Yang2009Newton}, we indeed find the solitary wave solutions of the form  (\ref{eq.heigen}), where the profiles are shown in Figure \ref{fig.nonconst}. In this simulation, the parameters are $\kappa(\bX)=\tanh(-X_2)$, $p_1=2, p_2=1$ and $\mu=-0.8$. We see that the modes are lump-like solutions lying on edge with certain symmetries.  This new type of nonlinear modes has not been reported yet in the literature.  As we see in the last subsection, modulated line modes can not persist in the nonlinear material. The nonlinear lump-like modes that we find could be the substitutes in the nonlinear media. This fully nonlinear mode is not moving due to the special choice of the ansatz. Actually, with a different choice of ansatz, we do find moving fully localized edge mode. This is beyond the scope of this work and left for our forthcoming work.

\begin{figure}[htbp]
  \centering
   \includegraphics[width=8cm]{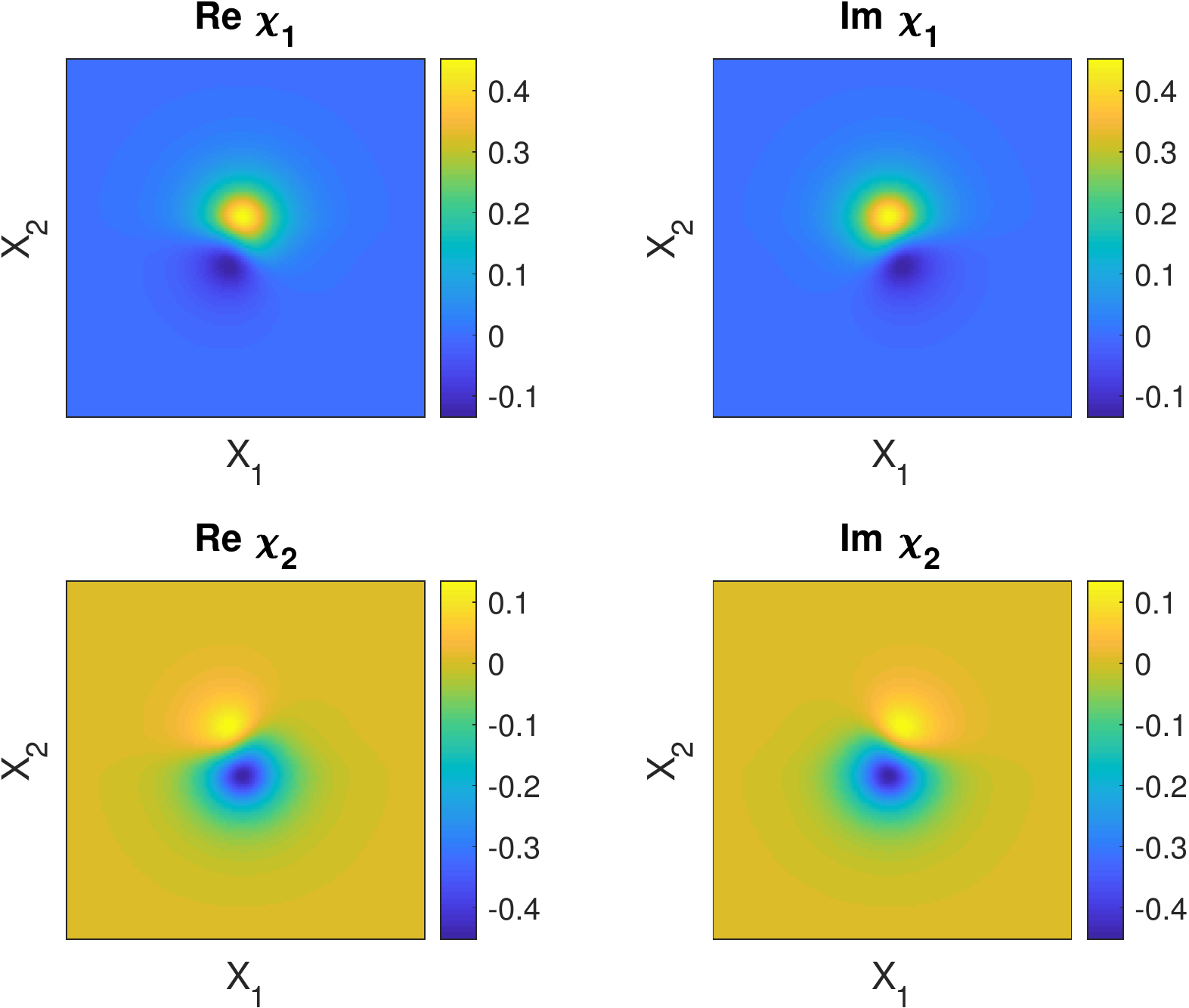}
  \caption{The profiles of the fully localized solitary wave solution (\ref{eq.heigen}) to the nonlinear Dirac equation (\ref{eq.Dirac2}) for $\mu = -0.8$. The nonlinear parameters are $p_1=2$ and $p_2=1$ and the mass is  $\kappa(\bX)=\tanh(-X_2)$.
   }
   \label{fig.nonconst}
\end{figure}

\section{Conclusions and discussions}\label{sec.conclusion}
Over the past ten years, there has been considerable interest in the wave dynamics in photonic meta-materials. Many novel propagating patterns have been produced and investigated. One of the key topics is the so-called topologically protected wave propagations in which the electromagnetic waves propagate along the designed path without any energy leaking to the bulks or traveling back even with strong defects. This robust wave propagations bring many potential applications.  In this paper,  nonlinear envelope dynamics of electromagnetic waves in nonlinear and weakly modulated honeycomb materials are studied.  By studying the envelope equation, we reveal the mechanism of some subtle wave patterns such as the topologically protected propagations. Different from the work in the existing literature, we directly study nonlinear Maxwell's equations and derive the nonlinear envelope equation.

We first investigate the spectrum of the Maxwell operator. With our characterization of honeycomb symmetries, we rigorously prove the existence of Dirac points which are conically singular points of the dispersion surfaces.  By the multi-scale perturbation theory, we derive the nonlinear dynamics of the envelope associated with the Dirac points in a weakly modulated honeycomb media. The reduced equation is a nonlinear Dirac equation with a spatially varying mass. We analyze and numerically simulate this equation to reveal the topologically protected edge states and their robust propagation. By including nonlinearity, we find the nonlinear edge states. Moreover, we report new lump-like modes which have not been found in the literature.  This new type of solitary solutions to the nonlinear Dirac equation may bring new features of the nonlinear materials and desire further investigations.

It is noted that some experimental realizations on topologically protected wave propagation are different from our physical setups, see for instance \cite{plotnik2013observation,poo2011experimental}. In their experiments, the edge states were observed at the interface between honeycomb materials and vacuum or air. This is referred to as a sharply terminated edge. Unfortunately, our current analysis, which utilizes a multi-scale analysis, does not apply for these physical setups. To the best of our knowledge, the mathematical analysis on such problems is mostly restricted in the tight-binding limit \cite{fefferman2018edge, ablowitz2017tight,ablowitz2013localized,ablowitz2013unified}.

\section{Acknowledgements}

This work was partially supported by the National Natural Science Foundation of China (grants $\#11871299$ and $\# 21877070$).


\bibliographystyle{unsrt}
\bibliography{references}

\begin{thebibliography}{10}

\bibitem{cheng2016robust}
Xiaojun Cheng, Camille Jouvaud, Xiang Ni, S~Hossein Mousavi, Azriel~Z Genack,
  and Alexander~B Khanikaev.
\newblock Robust reconfigurable electromagnetic pathways within a photonic
  topological insulator.
\newblock {\em Nature materials}, 2016.

\bibitem{HR:07}
FDM Haldane and S~Raghu.
\newblock Possible realization of directional optical waveguides in photonic
  crystals with broken time-reversal symmetry.
\newblock {\em Physical review letters}, 100(1):013904, 2008.

\bibitem{plotnik2013observation}
Yonatan Plotnik, Mikael~C Rechtsman, Daohong Song, Matthias Heinrich, Julia~M
  Zeuner, Stefan Nolte, Yaakov Lumer, Natalia Malkova, Jingjun Xu, Alexander
  Szameit, et~al.
\newblock Observation of unconventional edge states in ‘photonic graphene’.
\newblock {\em Nature materials}, 13(1):57, 2014.

\bibitem{poo2011experimental}
Yin Poo, Rui-xin Wu, Zhifang Lin, Yan Yang, and CT~Chan.
\newblock Experimental realization of self-guiding unidirectional
  electromagnetic edge states.
\newblock {\em Physical Review Letters}, 106(9):093903, 2011.

\bibitem{wang2009observation}
Zheng Wang, Yidong Chong, John~D Joannopoulos, and Marin Solja{\v{c}}i{\'c}.
\newblock Observation of unidirectional backscattering-immune topological
  electromagnetic states.
\newblock {\em Nature}, 461(7265):772, 2009.

\bibitem{ablowitz2010evolution}
Mark~J Ablowitz and Yi~Zhu.
\newblock Evolution of {B}loch-mode envelopes in two-dimensional generalized
  honeycomb lattices.
\newblock {\em Physical Review A}, 82(1):013840, 2010.

\bibitem{Fefferman2012Honeycomb}
Charles~L Fefferman and Michael~I Weinstein.
\newblock Honeycomb lattice potentials and {D}irac points.
\newblock {\em Journal of the American Mathematical Society}, 25(4):pags.
  1169--1220, 2012.

\bibitem{RMP-Graphene:09}
AH~Castro Neto, Francisco Guinea, Nuno~MR Peres, Kostya~S Novoselov, and
  Andre~K Geim.
\newblock The electronic properties of graphene.
\newblock {\em Reviews of modern physics}, 81(1):109, 2009.

\bibitem{joannopoulos2008molding}
John~D Joannopoulos, Steven~G Johnson, Joshua~N Winn, and Robert~D Meade.
\newblock Molding the flow of light.
\newblock {\em Princeton Univ. Press, Princeton, NJ [ua]}, 2008.

\bibitem{RevModPhys.82.3045}
M~Zahid Hasan and Charles~L Kane.
\newblock Colloquium: topological insulators.
\newblock {\em Reviews of modern physics}, 82(4):3045, 2010.

\bibitem{qi2011topological}
Xiao-Liang Qi and Shou-Cheng Zhang.
\newblock Topological insulators and superconductors.
\newblock {\em Reviews of Modern Physics}, 83(4):1057, 2011.

\bibitem{ablowitz2012nonlinear}
Mark~J Ablowitz and Yi~Zhu.
\newblock Nonlinear waves in shallow honeycomb lattices.
\newblock {\em SIAM Journal on Applied Mathematics}, 72(1):240--260, 2012.

\bibitem{ablowitz2013nonlinear}
Mark~J Ablowitz and Yi~Zhu.
\newblock Nonlinear wave packets in deformed honeycomb lattices.
\newblock {\em SIAM Journal on Applied Mathematics}, 73(6):1959--1979, 2013.

\bibitem{guo2019bloch}
Hailong Guo, Xu~Yang, and Yi~Zhu.
\newblock Bloch theory-based gradient recovery method for computing topological
  edge modes in photonic graphene.
\newblock {\em Journal of Computational Physics}, 379:403--420, 2019.

\bibitem{wallace1947band}
Philip~Richard Wallace.
\newblock The band theory of graphite.
\newblock {\em Physical Review}, 71(9):622, 1947.

\bibitem{ablowitz2011nonlinear}
Mark~J Ablowitz and Yi~Zhu.
\newblock Nonlinear diffraction in photonic graphene.
\newblock {\em Optics letters}, 36(19):3762--3764, 2011.

\bibitem{fefferman2018honeycomb}
Charles~L Fefferman, James~P Lee-Thorp, and Michael~I Weinstein.
\newblock Honeycomb {S}chr{\"o}dinger operators in the strong binding regime.
\newblock {\em Communications on Pure and Applied Mathematics},
  71(6):1178--1270, 2018.

\bibitem{lee2019elliptic}
James~P Lee-Thorp, Michael~I Weinstein, and Yi~Zhu.
\newblock Elliptic operators with honeycomb symmetry: {D}irac points, edge
  states and applications to photonic graphene.
\newblock {\em Archive for Rational Mechanics and Analysis}, 232(1):1--63,
  2019.

\bibitem{dohnal2013coupled}
Tom{\'a}s Dohnal and Willy D\"{o}rfler.
\newblock Coupled mode equation modeling for out-of-plane gap solitons in 2d
  photonic crystals.
\newblock {\em Multiscale Modeling \& Simulation}, 11(1):162--191, 2013.

\bibitem{curtis2015dynamics}
Christopher~W Curtis and Yi~Zhu.
\newblock Dynamics in-symmetric honeycomb lattices with nonlinearity.
\newblock {\em Studies in Applied Mathematics}, 135(2):139--170, 2015.

\bibitem{bal2017topological}
Guillaume Bal.
\newblock Topological protection of perturbed edge states.
\newblock {\em arXiv:1709.00605}, 2017.

\bibitem{bal2018continuous}
Guillaume Bal.
\newblock Continuous bulk and interface description of topological insulators.
\newblock {\em arXiv:1808.07908}, 2018.

\bibitem{xie2019wave}
Peng Xie and Yi~Zhu.
\newblock Wave packet dynamics in slowly modulated photonic graphene.
\newblock {\em Journal of Differential Equations}, 2019.

\bibitem{smirnova2019topological}
Daria Smirnova, Lev Smirnov, Daniel Leykam, and Yuri Kivshar.
\newblock Topological edge states and gap solitons in the nonlinear {D}irac
  model.
\newblock {\em arXiv:1904.07492}, 2019.

\bibitem{De_Nittis-Lein:14}
Giuseppe De~Nittis and Max Lein.
\newblock On the role of symmetries in the theory of photonic crystals.
\newblock {\em Annals of Physics}, 350:568--587, 2014.

\bibitem{Yang2010Nonlinear}
Jianke Yang.
\newblock {\em Nonlinear waves in integrable and nonintegrable systems}.
\newblock Society for Industrial and Applied Mathematics, 2010.

\bibitem{bao2017numerical}
Weizhu Bao, Yongyong Cai, Xiaowei Jia, and Qinglin Tang.
\newblock Numerical methods and comparison for the {D}irac equation in the
  nonrelativistic limit regime.
\newblock {\em Journal of Scientific Computing}, 71(3):1094--1134, 2017.

\bibitem{Yang2009Newton}
Jianke Yang.
\newblock Newton-conjugate-gradient methods for solitary wave computations.
\newblock {\em Journal of Computational Physics}, 228(18):7007--7024, 2009.

\bibitem{fefferman2018edge}
Charles~L Fefferman and Michael~I Weinstein.
\newblock Edge states of continuum {S}chr{\"o}dinger operators for sharply
  terminated honeycomb structures.
\newblock {\em arXiv:1810.03497}, 2018.

\bibitem{ablowitz2017tight}
Mark~J Ablowitz and Justin~T Cole.
\newblock Tight-binding methods for general longitudinally driven photonic
  lattices: Edge states and solitons.
\newblock {\em Physical Review A}, 96(4):043868, 2017.

\bibitem{ablowitz2013localized}
Mark~J Ablowitz, Christopher~W Curtis, and Yi~Zhu.
\newblock Localized nonlinear edge states in honeycomb lattices.
\newblock {\em Physical Review A}, 88(1):013850, 2013.

\bibitem{ablowitz2013unified}
MJ~Ablowitz and Y~Zhu.
\newblock Unified orbital description of the envelope dynamics in
  two-dimensional simple periodic lattices.
\newblock {\em Studies in Applied Mathematics}, 131(1):41--71, 2013.

\end{thebibliography}


\section*{Appendix}
\section{Dirac point in low contrast honeycomb media}
Theorem \ref{theo.conical} states that two-dimensional eigenspace $\mathcal{E}_{\omega_*}$ of $\MM_\omega$ at $\bK$ yields the existence of Dirac point $(\bK,\omega_{_D})$ as long as $\mathcal{E}_{\omega_*}\subset\L^2_{\bK,\tau}\oplus\L^2_{\bK,\overline{\tau}}$ and the non-degenerate condition $C_D>0$ hold. In the appendix, we show that the conditions ensuring the existence of Dirac points are satisfied in the low contrast honeycomb media.
To this end, we consider the material weight of the form
\begin{equation*}
    W_{\varepsilon}(\bx)= I_{3\times 3} + \varepsilon W^{1}(\bx),
\end{equation*}
where
  $W^1(\bx)=\Bigl(\begin{smallmatrix}
    A^1(\bx) & \mathbf{0}_{2\times 1}\\
     \mathbf{0}_{1\times 2}& a^1(\bx)\\
\end{smallmatrix}\Bigr)$ is a honeycomb material weight in Definition \ref{honeycomb_material}, and $\varepsilon>0$ is the perturbation constant.
We shall prove that this media has Dirac points in its dispersion band structure when $\varepsilon$ is sufficiently small.
Namely, we solve the following $\bK$ quasi-periodic eigenvalue problem perturbatively
\begin{equation}\label{eq.perturbeigen}
    \MM_{W_\varepsilon}\bPsi^\varepsilon:=W_{\varepsilon}(\bx)\LL \bPsi^\varepsilon = \omega^\varepsilon \bPsi^\varepsilon, \quad \bPsi^\varepsilon\in\L^2_{\bK}.
\end{equation}

First, we solve the non-perturbed eigenvalue problem (\ref{eq.perturbeigen}), i.e., $\varepsilon=0$ as
\begin{equation}\label{eq.nonperturbeigen}
    \MM_I\bPsi^0:=\LL \bPsi^0 = \omega^0 \bPsi^0, \quad \bPsi^0\in\L^2_{\bK}.
\end{equation}
The results are concluded in the proposition below.
\begin{prop}\label{prop.MI}
    The smallest positive eigenvalue of (\ref{eq.nonperturbeigen}) is $\omega^0=|\bK|$ with multiplicity three, and the corresponding eigenspace is
    \begin{equation*}
        \mathcal{E}^0_{_{|\bK|}}=\text{span}\big\{ \bPhi_1, \bPhi_2, \bPhi_3\big\},
        \end{equation*}
where

\begin{equation*}
        \bPhi_1=\begin{pmatrix}
            -\hat{\bK}\\
            1\\
        \end{pmatrix}e^{\ii\bK\cdot\bx} ,\quad \bPhi_2=\RRT\bPhi_1=\begin{pmatrix}
            -R\hat{\bK}\\
            1\\
        \end{pmatrix}e^{\ii R\bK\cdot\bx},\quad \bPhi_3=\RRT\bPhi_2=\begin{pmatrix}
            -R^2\hat{\bK}\\
            1\\
        \end{pmatrix}e^{\ii R^2\bK\cdot\bx},
     \end{equation*}

$\hat{\bK}=\bK/|\bK|$ and $R$ is the rotation matrix defined in (\ref{def.R}).
    Moreover,  the eigenspace can be decomposed as
    $\mathcal{E}^0_{_{|\bK|}}=\text{span}\{\bPsi^0_1\}\oplus\text{span}\{\bPsi^0_\tau\}\oplus\text{span}\{\bPsi^0_{\overline{\tau}}\}$, where
\begin{equation}
    \bPsi^0_\sigma = \frac{1}{\sqrt{6|\Omega|}}\big(\bPhi_1+\overline{\sigma}\bPhi_2+\sigma\bPhi_3\big)\in \L^2_{\bK,\sigma},\quad \sigma=1,\tau,\overline{\tau}.
    \label{eq:psiz}
\end{equation}
\end{prop}
Note that $\MM_I$ is a differential operator with constant coefficients. A direct calculation leads to the conclusion. Here we omit the detailed proof and refer the readers to \cite{Fefferman2012Honeycomb} for a similar calculation.

Next, we turn to the perturbed eigenvalue problem (\ref{eq.perturbeigen}). As $[\MM_{W_\varepsilon}, \RRT]=0$, we only need to solve this eigenvalue problem in the subspaces $\L^2_{\bK,\sigma}$, $\sigma=1,\,\tau,\,\overline{\tau}$  separately, i.e.,
\begin{equation}
    (\MM_I + \varepsilon\MM_{W^1})\bPsi_\sigma^\varepsilon = \omega_\sigma^\varepsilon\bPsi_\sigma^\varepsilon, \;
    \bPsi_\sigma^\varepsilon\in \mathbf{ \L^2_{\bK,\sigma}}.
    \label{eq:pert}
\end{equation}
By the perturbation theory, we shall prove the following theorem.

\begin{theo}\label{the.low}
    Let $\omega^0=|\bK|$ be the three-fold eigenvalue of $\MM_I$ and $\MM_{W_\varepsilon}$ be defined in (\ref{eq.perturbeigen}). Denote the Fourier coefficients of  $W^1$ as $W^{1}_{m_1,m_2}$, i.e.,
    \begin{equation*}
    W^{1}_{m_1,m_2}=\frac{1}{|\Omega|}\int_\Omega e^{-\ii(m_1\bk_1+m_2\bk_2)\cdot \textbf{y}}W^{1}(\textbf{y})\dd\textbf{y},\quad m_1,m_2\in\mathbb{Z}.
    \end{equation*}

    Assume the non-degeneracy condition holds
    \begin{equation*}
    \zeta^T W^{1}_{0,-1}\zeta
        \neq 0,
    \end{equation*}
    where $\zeta=\begin{pmatrix}
        \hat{\bK}\\
        1\\
    \end{pmatrix}$.  Then, there exist $\varepsilon_0 >0$, mappings $\varepsilon\rightarrow\omega_{_D}^\varepsilon$, and $\varepsilon\rightarrow\bPsi_\tau^\varepsilon$ ,
$\varepsilon\rightarrow\bPsi_{\overline{\tau}}^\varepsilon$ for  $0<\varepsilon<\varepsilon_0$, such that
     \begin{equation*}
            \omega_{_D}^\varepsilon=\omegaz+ \varepsilon\frac{\omegaz}{2}\big[\zeta^TW^{1}_{0,0}\zeta - \zeta^TW^{1}_{0,-1}\zeta \big]+O(\varepsilon^2)
        \end{equation*}
      is an eigenvalue of multiplicity two with eigenfunctions $\bPsi_\tau^\varepsilon\in \L^2_{\bK,\tau}$,$\bPsi_{\overline{\tau}}^\varepsilon\in \L^2_{\bK,\overline{\tau}}$ and
    the conical constant
    \begin{equation*}
C_D(\varepsilon) = \frac{1}{2}\Big|\mathcal{F}(\bPsi_1^\varepsilon,\bPsi_2^\varepsilon)\cdot\begin{pmatrix}
        1\\
        -\ii\\
    \end{pmatrix}\Big|+O(\varepsilon)    = \frac{1}{2}+O(\varepsilon)
\end{equation*}
satisfies that $C_D(\varepsilon)>0$ as $\varepsilon\ll 1$.
Thus $(\bK,\omega_D^\varepsilon)$ is a Dirac point by Definition \ref{theo.conical}.

\end{theo}

\begin{proof}

 To solve the eigenvalue problem (\ref{eq:pert}), we expand the eigenvalues and eigenfunctions as
\begin{equation}
    \omega_\sigma^\varepsilon = \omegaz+\varepsilon\omega_\sigma^{1,\varepsilon},\; \bPsi_\sigma^\varepsilon = \bPsisz + \varepsilon\bPsiso,
    \label{eq:perta}
\end{equation}
where $\la\bPsisz,\bPsiso\rai = 0$ and the inner product $\la.,.\rai$  is defined in $\L^2_{\bK}$ with the identity matrix weight $I$ as
\begin{equation*}
    \la\vf_1,\vf_2\rai = \int_\Omega\vf_1^*\vf_2\,\dd\bx.
\end{equation*}
Similar to the proof of Theorem \ref{theo.conical},  a Lyapunov-Schmidt reduction leads to that
\begin{equation}
\omega_\sigma^{1,\varepsilon} = \la\bPsisz, \MM_{W^1} \bPsisz\ra_I + O(\varepsilon^2), \quad \sigma=1,\tau,\overline{\tau}.
\end{equation}

Let $\omega_\sigma^{1,0} = \la\bPsisz,\MMo\bPsisz\rai$.
Recalling from Proposition \ref{prop.MI} that $\bPsisz = \frac{1}{\sqrt{6|\Omega|}}(\bPhi_1+\overline{\sigma}\bPhi_2+\sigma\bPhi_3)$ and $\MM_{W^1}=W^{1}\MM_I$, we directly calculate that
\begin{equation}\label{eq.omega1}
\begin{array}{lll}
\omega_\sigma^{1,0} &=& \la\bPsisz,\MMo\bPsisz\rai\\[0.3cm]
     &=& \frac{1}{6|\Omega|}\begin{pmatrix}
         1\\
             \overline{\sigma}\\
         \sigma\\
     \end{pmatrix}^*
    H \begin{pmatrix}
         1\\
             \overline{\sigma}\\
         \sigma\\
     \end{pmatrix},
\end{array}
\end{equation}
where we have used the fact that $\bPhi_j$, $j=1,2,3$ are the eigenfunctions of $\MM_I$ corresponding to the eigenvalue $\omega^0=|\bK|$, and the matrix $H=(H_{ij})_{3\times3}$ is given as follows
\begin{equation}
H_{ij}=\la\bPhi_i, \MM_{W^1}\bPhi_j\rai, \quad i,j=1,2,3.
\end{equation}
Evidently, $H$ is a Hermitian matrix by a similar argument in (\ref{sec.gap}).

By Proposition \ref{prop.inner}, $[\MM_{W^1},\RRT]=0$ and $\bPhi_3=\RRT\bPhi_2=\RRT^2\bPhi_1$, we obtain
\begin{equation*}
    \la\bPhi_1,\MMo\bPhi_1\rai = \la\RRT\bPhi_1,\MMo\RRT\bPhi_1\rai = \la\RRT^2\bPhi_1,\MMo\RRT^2\bPhi_1\rai.
\end{equation*}
Thus,
\begin{equation}\label{eq.h11}
    H_{11} =H_{22}=H_{33}=\la\bPhi_1,\MMo\bPhi_1\rai = \omegaz\la\bPhi_1,W^{1}\bPhi_1 \rai = \omegaz|\Omega|\zeta^TW^{1}_{0,0}\zeta,
\end{equation}
 where
$\zeta=\begin{pmatrix}
        \hat{\bK}\\
        1\\
    \end{pmatrix}$.

Similarly, one can derive that
\begin{equation}\label{eq.h22}
\begin{array}{lcl}
    H_{ij}=\omega^0|\Omega|\zeta^TW^{1}_{0,-1}\zeta, \quad i\neq j.
\end{array}
\end{equation}
Equations (\ref{eq.omega1})-(\ref{eq.h22}) yield
\begin{equation}\label{eq.omegadivde}
\begin{array}{lll}
\omega_\sigma^{1,0} &=& \la\bPsisz,\MMo\bPsisz\ra_I\\[0.3cm]
     &=& \frac{\omegaz}{2}\big[\zeta^TW^{1}_{0,0}\zeta+(\zeta^TW^{1}_{0,-1}\zeta)(\sigma+\overline{\sigma})\big],
\end{array}
\end{equation}
where $\omega_0=|\bK|$.

By $[\PPT\TTT,\MM_{W^\varepsilon}]=0$,
we have $\omega^\varepsilon_\tau=\omega^\varepsilon_{\overline{\tau}}:=\omega^\varepsilon_{_D}$. In other words, $\omega^\varepsilon_{_D}$ is an eigenvalue of multiplicity two with eigenfunctions $\bPsi_\tau^\varepsilon\in \L^2_{\bK,\tau}$,$\bPsi_{\overline{\tau}}^\varepsilon=(\PPT\TTT)\bPsi_\tau^\varepsilon\in \L^2_{\bK,\overline{\tau}}$. On the other hand, $\omega_n^\varepsilon:=\omega_1^\varepsilon\neq\omega^\varepsilon_{_D}$ by (\ref{eq.omegadivde}), which means $\omega_n^\varepsilon$ is an eigenvalue of multiplicity one with eigenfunction $\bPsi_1^\varepsilon\in \L^2_{\bK,1}$.

To bring an end to the proof, we shall verify the positivity of the conical constant $C_D(\varepsilon)$.
By (\ref{eq:perta}) and (\ref{eq.cf}), we have
\begin{equation*}
    C_D(\varepsilon) = \frac{1}{2}\Big|\mathcal{F}(\bPsi_\tau^\varepsilon,\bPsi_{\overline{\tau}}^\varepsilon)\cdot\begin{pmatrix}
        1\\
        -\ii\\
    \end{pmatrix}\Big|  =\frac{1}{2}\Big|\mathcal{F}(\bPsi_\tau^{0},\bPsi_{\overline{\tau}}^{0})\cdot\begin{pmatrix}
        1\\
        -\ii\\
    \end{pmatrix}\Big|    +O(\varepsilon).
\end{equation*}
Substituting (\ref{eq:psiz}) into (\ref{eq.deff}) implies
\begin{equation*}
\begin{array}{lll}
    \mathcal{F}(\bPsi_\tau^{0},\bPsi_{\overline{\tau}}^{0}) = -\frac{1}{3\omega_0}(1 + \overline{\tau}R + \tau R^2)\bK.
\end{array}
\end{equation*}
It follows that
\begin{equation*}
\begin{array}{lll}
C_D(\varepsilon) &=& \frac{1}{6\omega_0}\Big|(1 + \overline{\tau}R + \tau R^2)\bK\cdot\begin{pmatrix}
        1\\
        -\ii\\
    \end{pmatrix}\Big|+O(\varepsilon)\\
    &=& \frac{1}{2}+O(\varepsilon).
\end{array}
\end{equation*}
This completes the proof.
\end{proof}

\end{document}